\documentclass[12pt]{article}
\usepackage{custom_tex}
\usepackage{microtype}
\usepackage{enumitem}

\newcommand{\ob}[1]{{#1}_{\mathrm{o}}}
\newcommand{\thetaobs}{{\theta^{M_{\mathrm{o}}} } }

\title{Flexible Selective Inference with Flow-based Transport Maps}
\author[1]{Sifan Liu}
\affil[1]{Center for Computational Mathematics, Flatiron Institute}
\author[2]{Snigdha Panigrahi\thanks{The author gratefully acknowledges support by NSF CAREER Award DMS-2337882.}}
\affil[2]{Department of Statistics, University of Michigan}

\date{June 2025}

\begin{document}

\maketitle

\begin{abstract}
Data-carving methods perform selective inference by conditioning the distribution of data on the observed selection event. However, existing data-carving approaches typically require an analytically tractable characterization of the selection event. This paper introduces a new method that leverages tools from flow-based generative modeling to approximate a potentially complex conditional distribution, even when the underlying selection event lacks an analytical description---take, for example, the data-adaptive tuning of model parameters. The key idea is to learn a transport map that pushes forward a simple reference distribution to the conditional distribution given selection. This map is efficiently learned via a normalizing flow, without imposing any further restrictions on the nature of the selection event. Through extensive numerical experiments on both simulated and real data, we demonstrate that this method enables flexible selective inference by providing: (i) valid p-values and confidence sets for adaptively selected hypotheses and parameters, (ii) a closed-form expression for the conditional density function, enabling likelihood-based and quantile-based inference, and (iii) adjustments for intractable selection steps that can be easily integrated with existing methods designed to account for the tractable steps in a selection procedure involving multiple steps.
\end{abstract}

\section{Introduction}
\label{sec:intro}

A primary goal in data science is to extract models and formulate hypotheses from observed data. 
The practice of drawing inference from such selected models or hypotheses, without deferring the task to future data, is known as post-selection or selective inference. 
Whether it is basing inference on the most predictive model selected from a set of candidate models, testing hypotheses formulated after querying the data, or seeking inference in models where data-adaptive choices, such as the degrees of freedom or the model complexity parameter used in fitting, are involved—these are all examples of selective inference straight out of the data scientist's playbook.

A practical strategy to addressing selective inference in such problems is to account for the selection procedure by using a conditional distribution. 
This distribution is derived by conditioning on the selection event observed in the data, or on a suitable subset of that event.
A well-known example of this strategy is data-splitting, where the data is divided into two parts, one used for selection and the other reserved for inference. 
In this approach, inference is derived by conditioning on the data used for selection. 
This includes the sample splitting approach for independently and identically distributed data \citep{cox1975note}, as well as more recent approaches for certain parametric distributions, such as data-fission \citep{rasines2021splitting, dharamshi2025generalized, leiner2025data}, which involve dividing individual observations into two parts. 
When the two parts obtained by splitting the data are independent, the conditional distribution used for selective inference coincides with the unconditional distribution of the holdout data reserved exclusively for inference.

Conditioning on all the data used in selection, as is done in data-splitting, often amounts to conditioning on more information than is actually necessary.
Instead of adopting a data-splitting strategy, adjustments for the selection procedure can be made by conditioning on less information---ideally, only the event of selection itself. 
This conditioning approach, which discards only the information involved in selection, was introduced by \cite{fithian2014optimal,lee2016exact} for applications in variable selection.
In this paper, we refer to the class of conditional methods, which, unlike data-splitting, reuse data from selection, as  ``data-carving''.
When feasible, a data-carving approach can offer significant advantages over data-splitting methods, both in terms of selection accuracy and inferential power. 
This is because, in a data-splitting approach, selection tends to be less accurate when performed only on a subset of the data, and inference based solely on the holdout portion is less efficient at the same time, leading to tests with low power and unnecessarily wide confidence intervals.

\sloppy{Recent work have developed data-carving methods in various problems for regression, classification, and dimension reduction.}
This includes, for example, the work by \cite{le2022more, liu2023exact, panigrahi2023approximate, huang2023selective, gao2024selective, perry2024inference, panigrahi2024exact, pirenne2024parametric, bakshi2024inference}.
The adjustments based on a conditional distribution typically depend on an analytically tractable description of the selection event or, at the very least, utilize some knowledge of the geometry of this event.
In practice, however, the process of extracting models and hypotheses from data can be far more complex than the types of selection events these adjustments can handle.
For example, if a data scientist adaptively selects tuning parameters to control the complexity of model fit or balance the bias-variance tradeoff during fitting, describing the selection of these parameters is a notoriously difficult task.
While computing the necessary  adjustment for such selection events is beyond the reach of current data-carving methods, we develop a new method in this paper that computes inference from the complex, otherwise intractable conditional distribution based on the full data.

The key idea of our approach is to construct a transport map that transforms the conditional distribution of the relevant statistic to its pre-selection distribution, effectively acting as a debiasing transformation that removes the effect of selection. To obtain such a transport map, we adopt tools from flow-based generative modeling and learn this map in a data-driven manner. Importantly, our approach imposes no restrictions on the selection event or the algorithm leading to this event, applying even when the event lacks a tractable description. Instead, we only require that the selection procedure can be repeatedly applied to synthetically generated data, whose outputs are used as training data to fit the flow-based generative model.

Our main contributions are summarized below:
\begin{enumerate}[label=(\roman*)]
\item  We show how transport maps can be used to construct valid hypothesis tests and confidence sets for adaptively chosen parameters. We provide guarantees on the selective Type~$\Rom{1}$ error and selective coverage probability in terms of the accuracy of an approximate transport map.
\item The transport map directly yields an approximation to the conditional density function given the observed selection event. This allows for a range of inference procedures, such as those based on quantiles and conditional maximum likelihood estimation (MLE), as in \cite{panigrahi2022approximate}.
\item When the selection procedure involves multiple  steps, our method can be easily integrated with existing data-carving tools that adjust for some parts of the process with tractable descriptions.
For example, when applying the lasso for variable selection at a fixed regularization parameter, several types of data-carving tools exist to adjust for this type of selection. If the regularization parameter is chosen adaptively (e.g., via cross-validation), our method can account for this intractable tuning step and combine with these existing tools to ensure valid inference.
\end{enumerate}

The remainder of the paper is organized as follows.
In Section~\ref{sec: setup}, we present the conceptual framework for the data-carving approach, illustrate it with a concrete example, and review related work.
In Section~\ref{sec: method}, we describe how transport maps can be used to perform selective inference and provide selective inferential guarantees for our approach.
Section~\ref{sec: algorithm} presents our method for learning transport maps via normalizing flows.
Section~\ref{sec: extensions} discusses several extensions, including handling nuisance parameters and integrating our approach with existing data-carving tools.
In Section~\ref{sec: experiments}, we illustrate the application of our method on several examples that are beyond the reach of existing data-carving approaches, including an application to single-cell data analysis.
Section~\ref{sec:conc} concludes the paper with a discussion of future directions.

\section{Data-carving: framework and review}
\label{sec: setup}

In this section, we introduce the framework for data-carving methods, illustrate our proposed method with a first example, and provide a review of other related work.

\subsection{Framework}

\paragraph{Selection procedure}  Suppose a data scientist is provided with a dataset $D \in \calD$ and applies a selection procedure on this data to select a model. 
This procedure may involve additional randomness, independent of the data, such as from a data-splitting procedure or the addition of independent, external noise.
For example, this additional randomness may arise from train-test splits used during a model fitting procedure when tuning parameters are chosen via cross-validation, or from externally added noise to the selection algorithm to facilitate computationally efficient and more powerful selective inference, as done in \cite{tian2018selective, panigrahi2024exact, panigrahi2023approximate, huang2023selective}.
 We represent this external randomness by a randomization variable $W \in \calW$ with distribution $\bbQ$. If no additional randomness is involved in the selection procedure, we take $W = 0$. We denote by $\ob{D}$ and $\ob{W}$ the observed realizations of $D$ and $W$, respectively.

The selection procedure generates a model $M$ as a function of the data and randomization variable, which we denote as $M = \widehat{M}(D, W)$. The notion of a ``model'' is context-dependent but can be broadly understood as a collection of plausible data-generating distributions. For example, in regression settings, a model might correspond to a linear model involving a selected subset of predictors. In this case, the selection procedure determines which predictors to include in our model. Our framework here is quite general, and applies to a wide range of selection procedures, including those that can be formulated as algorithms. These include, for example, the choice of regularization parameters, degrees of freedom or other model complexity parameters, and the number of dimensions to reduce the data to—such as principal components—in dimension-reduction techniques, all of which are illustrated with numerical examples later in the paper.

\paragraph{Selective inference: goal and guarantees}

Given a model, a fundamental task is to draw inference about parameters based on observed data. However, when the model itself is selected in a data-dependent manner, classical inference procedures are no longer valid due to selection bias. To correct for this bias, data-carving methods base inference for the selected model $\ob{M} = \widehat{M}(\ob{D}, \ob{W})$ on distributions conditioned on the selection event.

Hereafter, we denote the target post-selection parameter in the selected model by $\thetaobs=\theta(\ob{M})$, where the superscript $\ob{M}$ emphasizes its dependence on the selected model.
Common inferential tasks include testing the null hypothesis $H_0: \theta^{\ob{M}} = \theta_0$ or constructing a confidence set for $\theta^{\ob{M}}$. For now, we assume that $\ob M$ is a parametric model fully specified by the parameter $\thetaobs$. We address the issue with nuisance parameters in Section~\ref{sec: extensions}.
Had the model $\ob M$ been fixed in advance, inference could proceed using a statistic  $T \sim \bbP_{\theta^{\ob{M}}}$, where $\bbP_{\theta^{\ob{M}}}$ is the ``pre-selection'' distribution of $T$. However, such a na{\"i}ve method fails to account for the data-dependent nature of the model and typically results in overly optimistic, invalid inference.

To adjust for selection, we seek to provide valid inference conditioned on the selected model, as outlined in \cite{fithian2014optimal}. 
For hypothesis testing, a test $\phi(T)$ is said to control the \emph{selective Type~$\Rom{1}$ error} at level $\alpha$ if
\begin{equation}
\label{selective:type1}
    \bbE_{\theta_0} \Big[\phi(T) \mid \widehat M=\ob{M} \Big]\leq\alpha, 
\end{equation}
where the expectation is taken under the joint distribution of $D\sim \bbP_{\theta_0}$ and $W\sim\bbQ$, conditional on the selection event $\widehat{M}(D, W) = \ob{M}$.
Similarly, a confidence set $C(T)$ for $\theta^{\ob{M}}$ achieves \emph{selective coverage probability} $1 - \alpha$ if
\begin{equation}
\label{selective:coverage}
\bbP_{\thetaobs} \left[\theta^{\ob{M}}\in C(T) \mid \widehat M=\ob{M} \right]\geq 1-\alpha.  
\end{equation}
By the law of iterated expectations, the selective inferential guarantees in \eqref{selective:type1} and \eqref{selective:coverage} imply unconditional validity, holding on average over the possible outcomes of $\widehat{M}$.

\paragraph{Conditional distribution under data-carving}

We now turn to the conditional distribution of $T$ under data-carving, which is derived by conditioning on the selection event. 
Let the density of $T$ under the pre-selection distribution $\bbP_{\thetaobs}$ be $p_{\thetaobs}$.
We denote the conditional distribution of $T\mid \{\widehat M(D,W)=\ob M \} $ under $\bbP_{\thetaobs}$ as $\bbP^*_{\thetaobs}$, with density given by
\begin{equation}
\label{carving:density}
p^*_{\theta^{\ob{M}}}(t) \propto p_{\theta^{\ob{M}}}(t) \times \PP[\thetaobs]{\widehat M =\ob{M}\mid T=t}.
\end{equation}
Here, $\PP[\thetaobs]{\widehat M =\ob{M}\mid T=t}$ represents the probability of selecting model $\ob M$ given the statistic $T=t$, marginalizing over any remaining randomness in the data $D$ and $W$. A formal derivation is provided in Appendix~\ref{app: conditional density}. The selective Type $\Rom{1}$ error and coverage probability guarantee in \eqref{selective:type1} and \eqref{selective:coverage} can be equivalently expressed as
\begin{align*}
    \EE[\bbP^*_{\theta_0} ]{\phi(T) } \leq \alpha\;\; \text{ and }\;\; \bbP^*_{\thetaobs} \left[\thetaobs \in C(T) \right] \geq 1-\alpha,
\end{align*}
respectively, using these notations.

Whether or not randomization is involved in the selection procedure, obtaining the selective distribution $\bbP^*_{\thetaobs}$ in \eqref{carving:density} in closed form can be infeasible unless the selection event $\{(D, W) \in \calD \times \calW : \widehat{M}(D, W) = \ob{M} \}$ admits a tractable description.
We now present a motivating data example in which the selection event is difficult to describe analytically, and preview how our proposed method enables valid selective inference without requiring an explicit description of the selection event.

\subsection{A first example and related work}

In the example below, we consider the problem of fitting a regression spline, where the number of knots is determined in a data-adaptive manner. Regression splines are commonly used to model nonlinear relationships between a response variable $y$ and a predictor variable $x$. The number of knots acts as a tuning parameter that controls the complexity of the model fit and is typically chosen using data-adaptive tools such as cross-validation (CV).


\begin{figure}
    \centering
    \includegraphics[width=\textwidth]{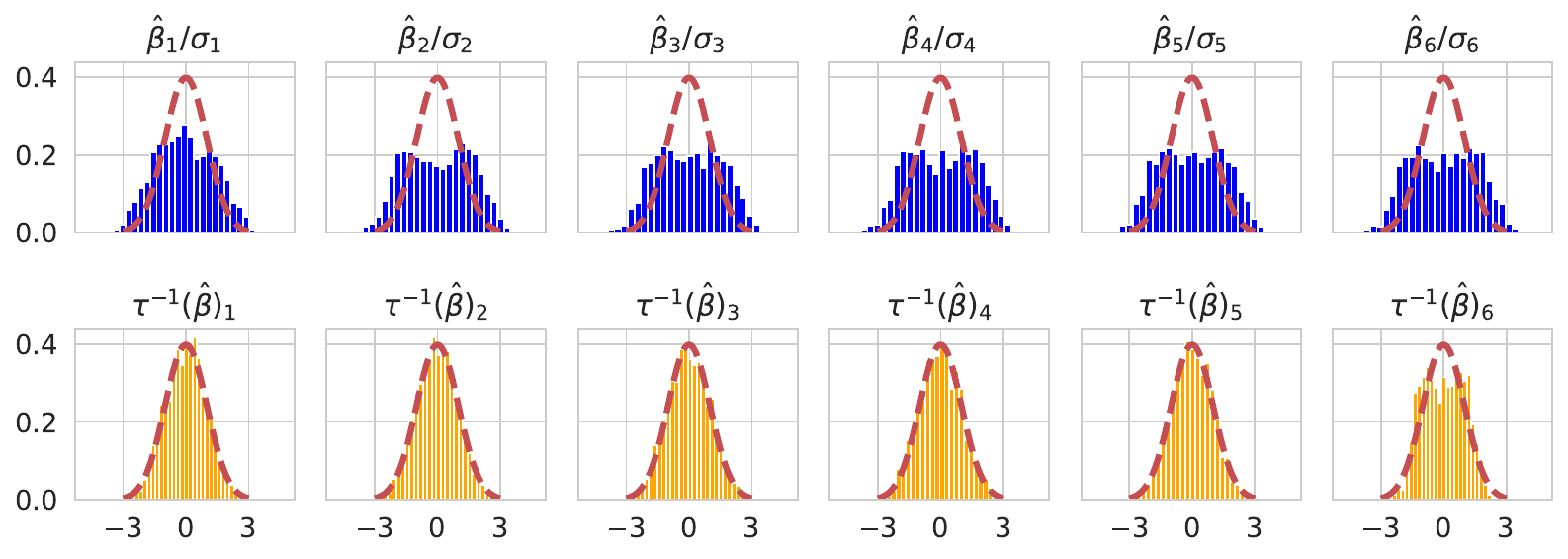}
    \caption{Top panel: histograms of $\hat\beta_j/\sigma_j$ conditioned on selecting 5 knots under the global null. Bottom panel: histograms of $\tau^{-1}(\hat\beta)_j$, where $\tau^{-1}$ is the learned transport map that pulls back the conditional distribution of $\hat\beta$ to the standard Gaussian distribution. The red dashed curves represent the density function of $\N(0,1)$. Due to the selection of knots via CV, prior to inference, the distribution of $\hat\beta_j/\sigma_j$ is distorted and is no longer $\N(0,1)$. The learned transport map $\tau$ pulls back the post-selection distribution of $\hat\beta$ to the standard Gaussian, yielding a new test statistic $\tau^{-1}(\hat\beta)_j$ for valid selective inference.}
    \label{fig: spline betahat hist}
\end{figure}

In particular, we perform regression using a natural cubic spline, with the number of knots selected through a 10-fold CV procedure. In our specific instance, $5$ knots are chosen, yielding 6 spline basis functions $\{b_1, \ldots, b_6\}$. Additionally, we include the constant function $b_0(x)=1$ to serve as an intercept. The selected nonlinear model $\ob{M}$ is given by
\begin{align}
    y_i = \sum_{k=0}^{6} \beta_k b_k(x_i) + \ep_i, \quad i \in \{1,2,\ldots, n\},
    \label{spline:selectedmodel}
\end{align}
where the errors $\ep_i\iid \N(0,\sigma^2)$. As our selective inferential task, we consider testing the global null hypothesis that the predictor is not associated with the response, i.e., $\beta_k=0$ for all $k=1,\ldots,6$.

If the selection of the number of knots is ignored in this problem, inference for the coefficients $\beta_k$ can be carried out using the least squares estimator $T = \hat{\beta}_{1:6}$ of $\beta_{1:6}$. Under the na{\"i}ve approach, the statistic $T$ follows a multivariate normal distribution $\mathcal{N}(0, \Sigma)$ under the global null, where $\Sigma = \sigma^2 ((X^\top X)^{-1})_{1:6, 1:6}$ and $X = (b_0(x), b_1(X), \ldots, b_K(X)) \in \R^{n \times 7}$ denotes the design matrix formed with the selected basis functions.

However, when we condition on the event that the selected number of knots is 5, the null distribution of $T=\hat{\beta}_{1:6}$ is no longer $\N(0,\Sigma)$.
To illustrate the effect of knot selection on the distribution of $\hat\beta$, we generate 2000 datasets under the global null model, in which the output of the CV procedure is 5 knots. In each replicate, we compute the least squares estimator $\hat{\beta}$. The top panel of Figure~\ref{fig: spline betahat hist} displays histograms of each coordinate of $\hat{\beta}$, standardized by its standard deviation $\sqrt{\Sigma_{k,k}}$. These empirical distributions show clear deviations from the standard normal, represented by the red dashed curves.


Using our method, we learn a pushforward transport map $\hat{\tau}$ from the pre-selection distribution $\bbP_{\theta_0}$ to the conditional distribution $\bbP^*_{\theta_0}$, which we formally define in Section \ref{sec: method}.
As a consequence, we have that if $T\sim \bbP^*_{\theta_0}$, then $\tau^{*-1}(T)\sim \N(0,\Sigma)$ (see Lemma \ref{lem:equivdefn}). The bottom panel of Figure~\ref{fig: spline betahat hist} displays histograms of each coordinate of $\tau^{*-1}(T)$, standardized by their corresponding standard deviations. These distributions closely match the standard normal distribution. This result demonstrates that, despite the analytical intractability of the selection event, the learned transport map in our approach successfully corrects for selection bias and ensures valid selective inference.


Our method can be seen related to approaches introduced for likelihood-free or simulation-based inference, which enable inference when the likelihood function is intractable, but generating data from the underlying model is feasible.
Work in this area include approximate Bayesian computation \citep{marin2012approximate}, synthetic likelihood approaches \citep{price2018bayesian} and simulation-based approaches \citep{xie2022repro, awan2024simulation}. 
Recent machine learning techniques that have made likelihood-free inference feasible include neural posterior estimation \citep{papamakarios2016fast}, neural likelihood estimation \citep{papamakarios2019sequential}, and ratio estimation methods \citep{cranmer2015approximating, thomas2022likelihood}.
Similar to the work in \cite{papamakarios2019sequential}, where an autoregressive flow is trained to approximate the likelihood of data given parameter values, our approach applies flow-based techniques to estimate a conditional likelihood function.

In the selective inference literature, \cite{liu2022black} learn the probability of a selection event given the data, which in turn facilitates a pivot for scalar-valued parameters.
In contrast, our method takes a different approach by using transport maps to learn the conditional distribution, and offers greater versatility in the range of inferential tasks it can perform. 
This includes, for example, inference for vector-valued parameters, joint likelihood-based inference, as well as the ability to combine with existing corrections for the more tractable parts of the selection procedure, which we discuss in Sections~\ref{sec: method} and \ref{sec: extensions}.
Furthermore, as shown in \cite{guglielmini2025selective}, the MLE from the conditional likelihood function---readily obtained from our current method---can be used to make inference on potentially complex functionals of the post-selection parameter when coupled with tools like the delta method and the bootstrap. 
In this sense, our approach may be extended to facilitate inference beyond just linear functionals. A different type of guarantee is offered by simultaneous methods for selective inference, such as those in \cite{berk2013valid, zrnic2024locally}, which aim at validity over a set of plausible targets instead of the one observed in our data. In contrast, the conditional guarantees in our work not only ensure valid inference for the selected model, but also enable data-adaptive inference, i.e., in cases where selection bias is either minimal or absent, inference from our approach matches with the na{\"i}ve approach.

\section{Data-carving using transport maps}
\label{sec: method}

In this section, we present the idea of using transport maps for performing selective inference.
We describe how to construct tests and confidence sets for an $\R^d$-valued parameter $\theta^{\ob{M}}$, when a transport map, or a sufficiently accurate approximation to it, is available. 
As discussed earlier in Section \ref{sec: setup}, we consider a  statistic $T$ whose pre-selection distribution is denoted by $\bbP_{\theta^{\ob{M}}}$.

\begin{definition}\label{def: pushforward}
Let $\tau^*_{\theta^{\ob{M}}} : \R^d \to \R^d$ be a diffeomorphism, i.e., a continuously differentiable and invertible map whose inverse $\tau_{\theta^{\ob{M}}}^{*-1}$ is also differentiable.
The map $\tau_{\theta^{\ob{M}}}^*$ is said to push forward the pre-selection distribution $\bbP_{\theta^{\ob{M}}}$ of $T$ to its conditional (post-selection) distribution $\bbP_{\theta^{\ob{M}}}^*$ if 
\begin{align*}
p_{\theta^{\ob{M}}}^*(t) = p_{\theta^{\ob{M}}}(\tau_{\theta^{\ob{M}}}^{*-1}(t))\cdot |\nabla \tau_{\theta^{\ob{M}}}^{*-1}(t)|,
\end{align*}
where $p_{\theta^{\ob{M}}}$ and $p_{\theta^{\ob{M}}}^*$ denote the densities of $\bbP_{\theta^{\ob{M}}}$ and $\bbP_{\theta^{\ob{M}}}^*$, respectively, and $|\nabla \tau_{\theta^{\ob{M}}}^{*-1}(t)|$ denotes the determinant of the Jacobian of the inverse map.
We denote this pushforward relationship by
\begin{align}
    \label{equ: pushforward}
    \tau_{\theta^{\ob{M}}}^* \# \bbP_{\theta^{\ob{M}}} = \bbP_{\theta^{\ob{M}}}^*,
\end{align}
\label{defn:TM}
and call $\tau_{\theta^{\ob{M}}}^*$ a transport map from $\bbP_{\theta^{\ob{M}}}$ to $\bbP_{\theta^{\ob{M}}}^*$.
\end{definition}

From Definition~\ref{def: pushforward}, the inverse map $\tau_{\theta^{\ob{M}}}^{*-1}$ can be interpreted as a \emph{pullback} from the conditional distribution to the pre-selection distribution. 
This is formalized in the following lemma.
\begin{lemma}
Suppose $\tau_{\theta^{\ob{M}}}^* : \R^d \to \R^d$ satisfies the pushforward relation~\eqref{equ: pushforward}.
Then, if $T\sim \bbP_{\theta^{\ob{M}}} ^*$, it follows that $\tau_{\theta^{\ob{M}}}^{*-1}(T)\sim \bbP_{\theta^{\ob{M}}}$.
\label{lem:equivdefn}
\end{lemma}
\begin{proof}
This result follows from the definition of $\tau_{\theta^{\ob{M}}}^*$.
\end{proof}


Put another way, the inverse map $\tau^{*-1}_{\theta^{\ob{M}} }$ acts as a debiasing transformation that corrects for the effect of selection. Under this transformation, the pulled-back statistic $\tau^{*-1}_{\theta^{\ob{M}} }(T)$ follows the pre-selection distribution. Consequently, any inference procedure that is valid under the pre-selection distribution can be directly applied to the pulled-back statistic. We elaborate on this idea in the remainder of this section.

\subsection{Hypothesis tests}

Given a valid test under the pre-selection distribution, a conditionally valid test---one that controls the selective Type~$\Rom{1}$ error as defined in \eqref{selective:type1}---can be obtained immediately by using the pullback of a transport map from $\bbP_{\theta^{\ob{M}}}$ to $\bbP_{\theta^{\ob{M}}}^*$.

\begin{thm}
Consider testing the null hypothesis $H_0:\theta^{\ob{M}}=\theta_0$. 
Let $\phi(T)$ be a level-$\alpha$ test under $\bbP_{\theta_0}$, i.e., $\EE[\bbP_{\theta_0}]{{\phi(T)}}\leq\alpha$. Suppose $\tau_{\theta_0}^*$ is a transport map satisfying the pushforward condition \eqref{equ: pushforward} when $\theta^{\ob{M}}=\theta_0$.
Then, $\phi \circ \tau_{\theta_0}^{*-1}(T)$ controls the selective Type~$\Rom{1}$ error at level $\alpha$, i.e., 
$$\EE[\bbP^*_{\theta_0}]{{\phi\circ\tau_{\theta_0}^{*-1}(T) }}\leq\alpha.$$
\label{thm:test}
\end{thm}
\begin{proof}
    Lemma \ref{lem:equivdefn} implies that
    $\widetilde{T} = \tau_{\theta_0}^{*-1}(T) \sim  \bbP_{\theta_0}$
    under the null hypothesis. Therefore, 
    \begin{align*}
        \EE[\bbP^*_{\theta_0}]{{\phi\circ\tau_{\theta_0}^{*-1}(T) }} = \EE[\bbP_{\theta_0}]{{\phi(\widetilde{T})}}\leq\alpha.
    \end{align*}
\end{proof}

In practice, the ideal transport map $\tau_{\theta_0}^*$ is typically not available. As described in Section~\ref{sec: algorithm}, our proposed method estimates an approximate map $\hat\tau_{\theta_0}$ to serve as a surrogate for the ideal map $\tau_{\theta_0}^*$. While $\hat\tau_{\theta_0}$ may not satisfy the exact pushforward condition in Equation~\eqref{equ: pushforward}, the induced distribution $\hat\tau_{\theta_0} \# \bbP_{\theta_0}$ can still closely approximate the target distribution $\bbP^*_{\theta_0}$, allowing for approximately valid selective inference.

To quantify the discrepancy between the approximate and target distributions, we consider the Kullback-Leibler (KL) divergence. For two probability measures $\bbP$ and $\bbQ$ defined on a common measurable space, the KL divergence is given by
\begin{align*}
    \kl{\bbP}{\bbQ} = \EE[\bbP]{\log \frac{\rd\bbP}{\rd\bbQ} }.
\end{align*}
The following theorem provides a quantitative bound on the selective Type $\Rom{1}$ error in terms of the KL divergence between $\bbP_{\theta_0}^*$ and $\hat\tau_{\theta_0}\# \bbP_{\theta_0} $.



\begin{thm}[Type $\Rom{1}$ error bound]\label{thm: error bound}
    Consider testing the null hypothesis $H_0:\theta^{\ob{M}}=\theta_0$. 
    Let $\phi(T)$ be a valid level-$\alpha$ test under $\bbP_{\theta_0}$, and let $\hat\tau_{\theta_0}$ be an invertible and differentiable map such that $\kl{\bbP_{\theta_0}^*}{\hat\tau_{\theta_0}\# \bbP_{\theta_0}}\leq\ep$. 
    Then the selective Type $\Rom 1$ error of the approximate test $\phi\circ\hat\tau_{\theta_0}^{-1}(T)$ is bounded by
    \begin{align*}
        \EE[\bbP^*_{\theta_0}]{\phi\circ \hat\tau_{\theta_0}^{-1}(T) } \leq \alpha + \sqrt{\ep/2}.
    \end{align*}
\end{thm}
\begin{proof}[Proof of Theorem~\ref{thm: error bound}]
    Since $\phi(T) \in [0, 1]$, the difference in expectations under two distributions can be bounded by their total variation distance:
    \begin{align*}
        \big|\EE[\hat\tau_{\theta_0}^{-1} \# \bbP^*_{\theta_0} ]{{\phi(T)} } - \EE[\bbP_{\theta_0}]{{\phi(T)} }\big|\leq \mathrm{TV}(\hat\tau_{\theta_0}^{-1}\# \bbP^*_{\theta_0}, \bbP_{\theta_0} ),
    \end{align*}
    where $\mathrm{TV}$ denotes the total variation distance between two probability distributions. 
    By Pinsker's inequality, 
    \begin{align*}
        \mathrm{TV}(\hat\tau_{\theta_0}^{-1}\# \bbP^*_{\theta_0}, \bbP_{\theta_0} )\leq \sqrt{\frac12 \kl{\hat\tau_{\theta_0}^{-1}\# \bbP^*_{\theta_0}}{\bbP_{\theta_0}} } =\sqrt{\frac12 \kl{\bbP^*_{\theta_0} }{\hat\tau_{\theta_0}\# \bbP_{\theta_0}} } \leq \sqrt{\ep/2},
    \end{align*}
    where the equality follows from the fact that $\hat\tau_{\theta_0}$ is invertible and differentiable, and that KL divergence is invariant under such transformations.
    Combined with the previous inequality, we obtain
    \begin{align*}
        \EE[\bbP^*_{\theta_0} ]{{\phi\circ \hat\tau_{\theta_0}^{-1}(T) } }=\EE[\hat\tau_{\theta_0}^{-1} \# \bbP^*_{\theta_0} ]{{\phi(T)} }  \leq \EE[\bbP_{\theta_0}]{{\phi(T)} } + \sqrt{\ep/2}\leq\alpha+\sqrt{\ep/2},
    \end{align*}
    where the last inequality holds since $\phi$ is a level-$\alpha$ test under $\bbP_{\theta_0}$.
\end{proof}

Theorem \ref{thm: error bound} establishes that if the KL divergence between the target measure $\bbP_{\theta_0}^*$ and $\hat\tau_{\theta_0}\# \bbP_{\theta_0} $ is at most $\ep$, then the selective Type~$\Rom{1}$ error is inflated by no more than $\sqrt{\ep/2}$.
This result provides both a theoretical guarantee for selective inference as well as a principled, data-driven criterion for learning the transport map, based on minimizing the KL divergence $\kl{\bbP^*_{\theta_0}} {\hat{\tau}_{\theta_0} \# \bbP_{\theta_0}}$. We describe this estimation procedure in Section~\ref{sec: algorithm}.

\subsection{Confidence sets}
\label{sec: confidence set}

Confidence sets for $\theta^{\ob{M}}$ with the selective coverage guarantee, as defined in \eqref{selective:coverage}, can be obtained by inverting the level-$\alpha$ tests established in Theorem~\ref{thm:test}. This leads to the following result.

\begin{proposition}
For each $\theta_0\in\Theta$, assume that the transport map $\tau^*_{\theta_0}$ satisfies the pushforward condition~\eqref{equ: pushforward}, and let $\phi_{\theta_0}\in\{0,1\}$ denote a level-$\alpha$ test for $H_0:\theta^{\ob{M}}=\theta_0$ under the distribution $\bbP_{\theta_0}$. 
Define the confidence set as
\begin{align}\label{equ: def confidence set}
    \calC^*(T) = \left\{\theta_0 \in\Theta: \phi_{\theta_0}\circ \tau_{\theta_0}^{*-1}(T) = 0 \right\}.
\end{align}
Then $\calC^*(T)$ achieves selective coverage at level $1-\alpha$, i.e., 
\begin{align*}
\bbP^*_{\theta^{\ob{M}}}\left[\theta^{\ob{M}}\in\calC^*(T) \right]\geq 1-\alpha.
\end{align*}
\label{prop:ci}
\end{proposition}

\begin{proof}
The proof follows by noting that
\begin{align*}
    \bbP^*_{\theta^{\ob{M}}}\left[\theta^{\ob{M}}\in\calC^*(T) \right] &= 1-\bbP^*_{\theta^{\ob{M}}}\left[\phi_{\theta^{\ob{M}}}\circ\tau_{\theta^{\ob{M}}}^{*-1}(T)=1 \right]\\
    & = 1-\bbP_{\theta^{\ob{M}}}[\phi_{\theta^{\ob{M}}}(T)=1 ]\geq 1-\alpha.
\end{align*}
\end{proof}

Analogous to Theorem~\ref{thm: error bound}, we establish a lower bound on the selective coverage probability of confidence sets constructed using approximate maps $\hat\tau_{\theta_0}$. Specifically, if each $\hat\tau_{\theta_0}$ induces a distribution such that the KL divergence from the exact selective distribution $\bbP^*_{\theta_0}$ is at most $\epsilon$ for all $\theta_0 \in \Theta$, then the resulting confidence sets achieve near-nominal selective coverage.

\begin{thm}[Coverage probability bound]
    Suppose that for all $\theta_0\in\Theta$, the approximate map $\hat\tau_{\theta_0}$ satisfies  $\kl{\bbP_{\theta_0} ^* }{\hat\tau_{\theta_0} \# \bbP_{\theta_0}  }\leq\ep$. Then the selective coverage probability of the confidence set $\widehat{\calC}$ satisfies
    \begin{align*}
       \bbP^*_{\theta^{\ob{M}}}\left[\theta^{\ob{M}}\in \widehat\calC(T)\right]\geq 1-\alpha - \sqrt{\ep/2},
    \end{align*}
    where $\widehat{\calC}$ is defined analogously to $\calC^*$ in \eqref{equ: def confidence set}, with $\tau^*_{\theta_0}$ replaced by $\hat\tau_{\theta_0}$.
    \label{thm:coverage:prob}
\end{thm}

\begin{proof}
The proof follows directly from Theorem~\ref{thm: error bound}.
\end{proof}

\subsection{Inference based on conditional density}
\label{sec: conditional density}

A transport map $\tau^*_{\theta^{\ob{M}}}$ in \eqref{equ: pushforward} also provides an expression for the selective density of the statistic $T$ under the conditional distribution $\bbP^*_{\theta^{\ob{M}}}$, as stated in Definition~\ref{def: pushforward}. Accordingly, an approximate map $\hat{\tau}_{\theta^{\ob{M}}} \approx \tau^*_{\theta^{\ob{M}}}$ yields an approximate selective density of the form
\begin{align}\label{push:forward:density}
    q^*_{\thetaobs}(t) = p_{\thetaobs}(\hat\tau^{-1}_{\thetaobs}(t) ) \cdot |\nabla \hat\tau^{-1}_{\thetaobs}(t) |.
\end{align}
This approximate selective density can now serve as the basis for conducting selective inference. In what follows, we discuss two approaches that directly use $q^*_{\theta^{\ob{M}}}$ for selective inference.

\paragraph{Selective MLE}
The approximate selective density~\eqref{push:forward:density} gives rise to the following negative log-likelihood function
\begin{align}\label{equ: selective log likelihood}
    \ell^*(\theta) := -\log q^*_{\theta}(T)=-\log p_{\theta}(\hat\tau_{\theta}^{-1}(T) ) - \log |\hat\nabla \tau_{\theta}^{-1}( T) |.
\end{align}
Minimizing the negative log-likelihood $\ell^*(\theta)$ yields a natural point estimate, the selective MLE, defined as
\begin{equation}
    \widehat\theta^{\ob{M}}_{\text{mle}} = \underset{\theta\in\Theta}{\argmin}\;  \ell^*(\theta).
\label{cond:MLE}    
\end{equation}

In the presence of external randomization, \cite{panigrahi2022approximate} obtained selective inference for $\theta^{\ob{M}}$ based on approximate normality of the selective MLE.
This approach constructs Wald-type intervals centered at the selective MLE, with variance estimated using the observed Fisher information matrix, both of which are computed from the conditional likelihood.
We can apply the same procedure here, using the selective likelihood density defined in Equation~\eqref{push:forward:density}.

As an example, suppose we wish to construct an equi-tailed confidence interval for the $j^{\text{th}}$ component of $\theta^{\ob{M}}$. 
We compute the selective MLE $\widehat\theta^{\ob{M}}_{\text{mle}}$ as in \eqref{cond:MLE}, and estimate the Fisher information matrix as 
$$
\widehat\calI=\calI(\widehat\theta^{\ob{M}}_{\text{mle}} ) = \nabla^2_\theta \; \ell^*(\theta)\, \Big\lvert _{\widehat\theta^{\ob{M}}_{\text{mle}}}.
$$
Based on the approximate normality of $\widehat\theta^{\ob{M}}_{\text{mle}}$, a $(1-\alpha)$ confidence interval for $\theta_j^{M_{\mathrm{o}}}$ is given by
\begin{align*}
   [\widehat\theta^{\ob{M}}_{\text{mle}}]_{j} \pm  \sqrt{[\widehat\calI^{-1}]_{j,j}} \cdot z_{1-\alpha/2},
\end{align*}
where $[\widehat\theta^{\ob{M}}_{\text{mle}}]_{j}$ is the $j^{\text{th}}$ component of the selective MLE, $[\widehat\calI^{-1}]_{j,j}$ is the $(j,j)^{\text{th}}$ component of the inverse Fisher information matrix $\widehat\calI^{-1}$, and $z_{1-\alpha/2}$ is the $(1-\alpha/2)$ quantile of the standard normal distribution.

\paragraph{Quantile-based inference}
If $\theta^{\ob{M}}$ is a scalar-valued parameter, one can perform selective inference based on the quantiles of the conditional density~\eqref{push:forward:density}. Concretely, suppose we want to test the null hypothesis $H_0: \theta^{\ob{M}} = \theta_0$ against one-sided or two-sided alternatives. 
In this case, valid p-values controlling the selective Type~$\Rom{1}$ error can be computed as
\begin{align*}
    \int_{-\infty}^T q_{\theta_0}^*(t)\rd t,\quad \int_T^\infty q_{\theta_0}^*(t)\rd t,\quad 2\cdot \min\left(\int_{-\infty}^T q_{\theta_0}^*(t)\rd t,\, \int_T^\infty q_{\theta_0}^*(t)\rd t\right),
\end{align*}
corresponding to left-tailed, right-tailed, and two-sided tests, respectively.

Analogous to Theorem~\ref{thm: error bound}, if $\hat\tau_{\theta_0}\# \bbP_{\theta_0}$ is close to $\bbP^*_{\theta_0}$ in KL divergence, then the p-values based on the corresponding densities are also close. Specifically, we have the bound:
\begin{align*}
    \bigg|\int_{-\infty}^T p_{\theta_0}^*(t)\rd t - \int_{-\infty}^T q^*_{\theta_0}(t)\rd t \bigg| \leq \mathrm{TV}(\bbP_{\theta_0}^*, \hat\tau_{\theta_0}\# \bbP_{\theta_0})\leq \sqrt{\frac12\kl{\bbP_{\theta_0}^*}{\hat\tau_{\theta_0}\# \bbP_{\theta_0}}}.
\end{align*}
If $\kl{\bbP_{\theta_0}^*}{\hat\tau_{\theta_0}\# \bbP_{\theta_0}}\leq\ep$, the approximate p-values computed using $q^*_{\theta_0}$ differ from the exact p-values by at most $\sqrt{\epsilon/2}$.


\section{Learning transport maps using normalizing flows}
\label{sec: algorithm}

We now describe our approach for learning the transport map $\tau^*_{\thetaobs}$ that pushes forward the pre-selection distribution $\bbP_{\thetaobs}$ to the conditional distribution $\bbP^*_{\thetaobs}$. We first present our method in the context of hypothesis testing, and subsequently extend it to constructing confidence sets and obtaining an approximation to the selective density.

\subsection{Minimizing the KL divergence}

Consider the problem of testing the null hypothesis $H_0:\theta^{\ob{M}}=\theta_0$.
As discussed in Section~\ref{sec: method}, an approximate transport map can be obtained by minimizing the KL divergence between the target conditional distribution $\bbP_{\theta_0}^*$ and the pushforward distribution $\hat\tau_{\theta_0} \# \bbP_{\theta_0}$. 
As shown in Theorem \ref{thm: error bound}, minimizing the KL divergence between the two distributions enables us to control the selective Type $\Rom 1$ error. This makes KL minimization not only a natural learning objective but also a theoretically justified approach for performing valid selective inference.

Using the formula for the density of $\hat\tau_{\theta_0}\# \bbP_{\theta_0}$ in \eqref{push:forward:density}, we note that the KL divergence can be expressed as
\begin{align*}
\kl{\bbP_{\theta_0}^*}{\hat\tau_{\theta_0} \# \bbP_{\theta_0}} 
&=\EE[\bbP_{\theta_0}^* ]{\log p_{\theta_0}^*(T) - \log p_{\theta_0} (\hat\tau_{\theta_0}^{-1}(T) ) - \log|\nabla \hat\tau_{\theta_0}^{-1}(T) |  }.
\end{align*}
To evaluate the expectation, we draw training samples $T^{(b)}$ for $1\leq b\leq B$ from the target distribution $\bbP^*_{\theta_0}$. Ignoring the constant term which doesn't depend on the transport map, we arrive at the empirical version of the objective function
\begin{align}
\label{equ: kl objective}
    \frac1B\sum_{b=1}^B - \log p_{\theta_0} (\hat\tau_{\theta_0}^{-1}(T^{(b)})) - \log|\nabla \hat\tau_{\theta_0}^{-1}(T^{(b)})|.
\end{align}
However, directly optimizing over arbitrary diffeomorphisms is generally intractable. In practice, we adopt a variational approach to solve this problem by parameterizing the transport map within a flexible class of functions, which we describe in the following section.

\subsection{Parametrizing transport maps via normalizing flows}
\label{sec: NF}

We propose to parametrize the transport maps using normalizing flows, a class of generative models that transform a simple reference distribution into a complex target distribution through a sequence of invertible, differentiable maps. In our setting, the reference distribution is the pre-selection distribution $\bbP_{\theta_0}$, and our target distribution is the conditional (post-selection) distribution $\bbP^*_{\theta_0}$.

Since the objective function~\eqref{equ: kl objective} involves only the inverse map, we directly parameterize the inverse transport map $\hat{\tau}^{-1}_{\theta_0}$ using a normalizing flow family $\{\eta_{\theta_0}(\cdot;\psi) : \psi \in \Psi\}$, where $\psi \in \Psi$ denotes the (unknown) flow parameters to be optimized.
This leads to the following optimization problem 
\begin{equation}
\widehat{\psi}= \underset{\psi\in\Psi}{\text{argmin}} \ \frac1B\sum_{b=1}^B - \log p_{\theta_0} (\eta_{\theta_0}(T^{(b)}; \psi)) - \log|\nabla\eta_{\theta_0}(T^{(b)}; \psi)|.
\label{opt:problem}
\end{equation}
The solution to \eqref{opt:problem} yields the approximate inverse transport map $\widehat{\tau}^{-1}_{\theta_0}(\cdot):=\eta_{\theta_0}(\cdot; \widehat{\psi})$.

Various flow architectures have been proposed in the literature on normalizing flows, offering different trade-offs between expressiveness and computational efficiency. In all examples presented in this paper, we employ the real-valued non-volume preserving (RealNVP) flows in \cite{dinh2017density}. This choice ensures that the Jacobian of the transformation—appearing in the second term of the objective in Equation~\eqref{opt:problem}—is a triangular matrix, allowing its determinant to be computed efficiently. Implementation details and parameterization of the RealNVP flows used in our experiments are provided in Appendix~\ref{sec: experiment details}.

To summarize, our approach to learning the transport map involves two main steps.
In Step 1, we generate training data $\{T^{(b)} \sim \bbP_{\theta_0}^*,\; 1\leq b\leq B\}$ from the conditional distribution under the null. To do so, we employ rejection sampling that is carried out as follows:
\begin{enumerate}
    \item[(i)] Generate $D^{(b)}\sim \bbP_{\theta_0}$ and $W^{(b)}\sim \bbQ$.
    \item[(ii)] Apply the selection algorithm $\widehat M$ to $(D^{(b)}, W^{(b)})$.
    \item[(iii)] If $\widehat M(D^{(b)}, W^{(b)})=\ob{M}$, accept $T^{(b)}=T(D^{(b)})$ as a sample.
\end{enumerate}

In Step 2, we parametrize the inverse transport map using the RealNVP flow, as described above, and solve the optimization problem in \eqref{opt:problem} over the flow parameter space via gradient descent. By combining the flow-based (inverse) transport map $\hat\tau_{\theta_0}^{-1}$ with Theorem~\ref{thm:test}, we obtain a valid testing procedure. Our approach is summarized in Figure~\ref{fig: pipeline}.

\begin{figure}[ht]
  \centering
  \setstretch{1}
  \begin{tikzpicture}[node distance=5.4cm, auto]
      \tikzstyle{block} = [rectangle, draw, 
                            minimum width=4.8cm, 
                            minimum height=2.5cm, 
                            text width=4.6cm,
                            align=center, font=\small]
      \tikzstyle{arrow} = [double equal sign distance, -implies] 

      \node[block] (A) { \textbf{Generating training samples} \\ \vspace{.1cm} Draw $(D^{(b)},W^{(b)})\sim (\bbP_{\theta_0}\times \bbQ)$; if $\widehat M(D^{(b)}, W^{(b)})=\ob{M}$, accept $T^{(b)}=T(D^{(b)})$ as a training sample. };

      \node[block, right of=A] (B) {\textbf{Learning transport map}\\ \vspace{.1cm} Optimizing flow parameter $\psi$ by solving problem~\eqref{opt:problem}.   };

      \node[block, right of=B] (C) {\textbf{Selective inference} \\ \vspace{.1cm} Test $H_0:\theta^{\ob{M}}=\theta_0$ by $\phi\circ \hat\tau_{\theta_0}^{-1}(T)$, where $\phi$ is a valid level $\alpha$-test under $\bbP_{\theta_0}$.};
      \draw[arrow] (A) -- (B);
      \draw[arrow] (B) -- (C);
  \end{tikzpicture}
  \caption{Pipeline of the proposed approach for hypothesis testing.}
  \label{fig: pipeline}
\end{figure}
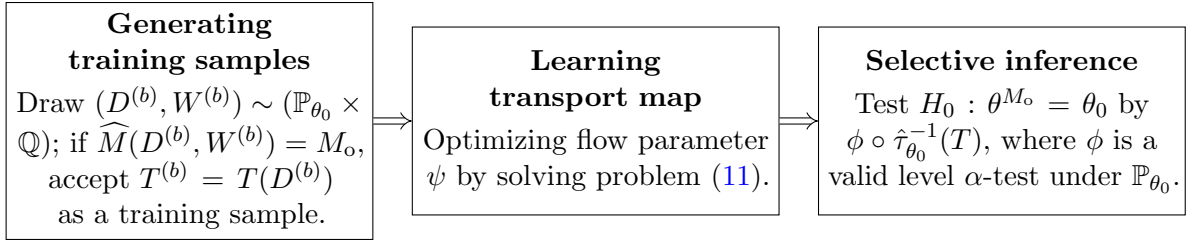

\subsection{Conditional normalizing flows for efficient inference}
\label{sec: conditional NF}

As discussed in Section \ref{sec: method}, confidence sets for the parameter $\theta^{\ob{M}}$ can be constructed by inverting hypothesis tests. 
We have described how to learn a transport map $\tau_{\theta_0}^*$ for a fixed parameter value $\theta_0\in\Theta$ by generating data from $\bbP_{\theta_0}^*$ and training a normalizing flow. 
However, repeating this process for many different parameter values can be computationally expensive. 
To reduce this cost, we propose a conditional normalizing flow approach that learns a single family of transport maps $\{\hat{\tau}_{\theta},\, \theta \in \Theta\}$, indexed by $\theta$, enabling efficient construction of confidence sets and more generally, estimation of the conditional density $p^*_\theta(T)$ across a range of parameter values.

Revisiting Theorem \ref{thm:coverage:prob}, one can construct confidence sets with near-nominal coverage by minimizing the KL divergence between the pushforward measure induced by $\hat\tau_{\theta}$ and the target distribution $\bbP^*_\theta$, for all $\theta \in \Theta$.
Toward this goal, we minimize the expected KL divergence:
\begin{align*}
    \EE[\theta^{\ob{M}}\sim \pi]{\kl{\bbP_{\theta^{\ob{M}}}^*}{\hat\tau_{\theta^{\ob{M}}} \# \bbP_{\theta^{\ob{M}}}}},
\end{align*}
where $\pi$ is a distribution supported on $\Theta$. In practice, this distribution is user-specified to generate plausible parameter values from the support set $\Theta$.

To optimize this objective, we consider a conditional normalizing flow family  $\{\eta(\cdot, \cdot; \psi),\, \psi \in \Psi\}$, where, for fixed $\psi$, $\eta(T, \theta; \psi)$ takes as input both the statistic $T$ and the parameter value $\theta$, with $\psi$ denoting the shared parameters of the conditional flow.
We train this model using paired samples $\{(T^{(b)}, \theta^{(b)}),\, 1\leq b\leq B\}$, where $\theta^{(b)} \sim \pi, \, T^{(b)} \mid \theta^{(b)} \sim \bbP_{\theta^{(b)}}^*$.
In our experiments, the distribution $\pi$ is chosen to be a multivariate Gaussian distribution, with details provided in Section \ref{sec: experiments}.
Given these training samples, we solve the optimization problem
\begin{equation}
\widehat{\psi}= \underset{\psi\in\Psi}{\text{argmin}} \ \frac1B\sum_{b=1}^B - \log p_{\theta^{(b)}} (\eta(T^{(b)}, \theta^{(b)}; \psi)) - \log|\nabla\eta(T^{(b)}, \theta^{(b)}; \psi)|,
\label{opt:problem:gen}
\end{equation}
yielding an estimate $\widehat{\psi}$ and the approximate inverse transport map $\widehat{\tau}^{-1}_{\theta^{\ob{M}}}(\cdot):=\eta(\cdot, \theta^{\ob{M}}; \widehat\psi)$, which can be used to construct confidence sets of $\thetaobs$ as described in Section~\ref{sec: confidence set}. The procedure for constructing confidence interval is summarized in Figure~\ref{fig: pipeline CI}.

\begin{figure}[ht]
  \centering
  \setstretch{1}
  \begin{tikzpicture}[node distance=5.4cm, auto]
      \tikzstyle{block} = [rectangle, draw, 
                            minimum width=4.8cm, 
                            minimum height=2.5cm, 
                            text width=4.6cm,
                            align=center, font=\small]
      \tikzstyle{arrow} = [double equal sign distance, -implies] 

      \node[block] (A) { \textbf{Generating training samples} \\ \vspace{.1cm} 
      Draw $\theta^{(b)}\sim\pi$ and $(D^{(b)},W^{(b)})\sim (\bbP_{\theta^{(b)}}\times \bbQ)$; if $\widehat M(D^{(b)}, W^{(b)})=\ob{M}$, accept $ (T(D^{(b)}), \theta^{(b)})$ as a training sample. };

      \node[block, right of=A] (B) {\textbf{Learning transport map}\\ \vspace{.1cm} Optimize the conditional flow parameter $\psi$ by solving problem~\eqref{opt:problem:gen}.   };

      \node[block, right of=B] (C) {\textbf{Selective inference} \\ \vspace{.1cm} 
      Construct confidence set for $\thetaobs$ as $\{\theta\in\Theta: \phi_{\theta}\circ \hat\tau^{-1}_{\theta}(T) = 0 \}$
      where $\phi_\theta$ is a valid level $\alpha$-test under $\bbP_{\theta}$.};
      \draw[arrow] (A) -- (B);
      \draw[arrow] (B) -- (C);
  \end{tikzpicture}
  \caption{Pipeline of the proposed approach for constructing confidence intervals.}
  \label{fig: pipeline CI}
\end{figure}
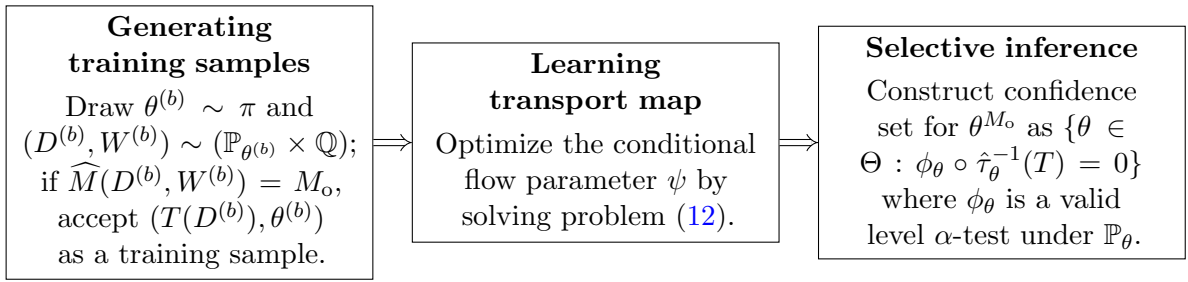

Additionally, the conditional normalizing flow provides a surrogate for the selective likelihood function defined in \eqref{equ: selective log likelihood}.
By substituting the learned inverse map $\widehat{\tau}^{-1}_{\theta}$ into the likelihood function, we obtain an approximation to the selective likelihood, which then facilitates maximum likelihood inference using the procedure described in Section~\ref{sec: conditional density}, or quantile-based inference for scalar-valued parameters.

In our description of the method, the transport maps use the pre-selection distribution as the reference distribution. However, as noted in the following remark, the reference distribution need not be the pre-selection distribution.
\begin{rmk}[Choice of reference distribution]\label{rmk: reference measure}
When training the normalizing flow, the reference distribution does not have to be the pre-selection distribution $\bbP_{\theta^{\ob{M}}}$. It can, in fact, be any probability distribution for which the log-density can be computed in a tractable form. In the variational inference and normalizing flow literature, the standard normal distribution $\N(0,I_d)$ is a common choice for the reference measure. However, using $\bbP_{\theta^{\ob{M}}}$ instead of $\N(0,I_d)$ may offer practical advantages. Since $\bbP_{\theta^{\ob{M}}}$ may be closer to the target distribution $\bbP_{\theta^{\ob{M}}}^*$ when the effect of selection is less, a simpler transformation may suffice to approximate the target in such examples, thereby making the transport map easier to learn and potentially improving training efficiency.
\end{rmk}


\section{Extensions}
\label{sec: extensions}

In this section, we describe extensions of our method to accommodate different scenarios. In particular, we discuss how to condition on extra information from the selection procedure to address issues such as nuisance parameters, and how to combine our method with existing methods for selective inference that provide corrections for parts of the selection procedure with analytically tractable selection events.

\subsection{Conditioning on extra information}




In selective inference, it can be beneficial to condition not only on the selection event $\{\widehat{M} = \ob{M}\}$, but also on additional statistics. Let $A = A(D)$ denote such a statistic. Then any inference procedure that is valid conditional on $\{\widehat{M} = \ob{M}, A = \ob{A}\}$ remains valid when conditioning only on $\{\widehat{M} = \ob{M}\}$, by the law of iterated expectation.

The additional conditioning on $A$ typically serves the following purposes. First, it can help eliminate nuisance parameters from the model, as done in the construction of uniformly most powerful unbiased (UMPU) tests for exponential families \citep{fithian2014optimal} or in forming the selective likelihood function within the M-estimation framework \citep{huang2023selective}. The approximate conditional density $q^*_{\theta}$~\eqref{push:forward:density} obtained from our approach can be easily used to implement this strategy for handling nuisance parameters.
Second, conditioning on additional information can improve the efficiency of the rejection sampling that we described in Section~\ref{sec: algorithm}. Since the selection event $\{\widehat M=\ob{M} \}$ holds for the observed data, additionally conditioning on part of the data, in this case $A=\ob A$, can increase the chance of re-selecting $\ob M$, thereby improving acceptance rates.
While marginalizing over $A$ may be an option, conditioning on its observed value often provides a more computationally efficient alternative, though possibly at the expense of some statistical power.

\sloppy{In our framework, conditioning on extra information requires only minor modification to our learning approach in Section \ref{sec: algorithm}. When generating the training samples, we draw datasets $D^{(b)}\sim \bbP_{\theta}(\cdot \mid A=\ob{A})$, and then accept a sample $T^{(b)}$ if $\widehat M(D^{(b)}, W^{(b)} ) = \ob {M}$.
Here is an example illustrating how this works.}
In a fixed $X$-regression setting, consider the linear model $y\sim \N(X_{\ob{M}}\beta_{\ob{M}}, \sigma^2 I_n)$, where the design matrix $X_{\ob{M}}$ is determined by the selection procedure $\widehat M$. In this case, additionally conditioning on the statistic $A= (I_n - X_{\ob{M}} (X_{\ob{M}}\tran X_{\ob{M}})^{-1} X_{\ob{M}}\tran )y$, which is independent of the least squares estimator, is likely to increase the probability of re-selecting the same model $\ob M$. Since $y$ follows a Gaussian distribution and $A$ is a linear function of $y$, simulating from $y^{(b)}\mid A=\ob{A}$ is straightforward. 
Similarly, when additional randomization is involved, such as a train-test split used during model fitting, we can condition on these randomization variables, in this case the random choice of split, at the time of inference.
Once the training samples are generated under this conditional distribution, the subsequent steps---training the normalizing flow and conducting inference---proceed exactly as before.

\subsection{Integrating with existing selective inference methods}
\label{sec: multiple queries}


In practice, data analysis pipelines may involve multiple stages of selection. Our proposed method is particularly well-suited for handling selection steps with intractable selection events, and at the same time, it can be easily combined with existing selective inference methods that apply to the tractable steps of the selection procedure.

As a motivating example, consider the well-studied problem of variable selection via the lasso. A long line of prior work, including \cite{lee2016exact, liu2023exact, panigrahi2024exact, kivaranovic2021length}, have studied selective inference methods after lasso selection with a fixed regularization parameter $\lambda$. However, in most applications, $\lambda$ is selected in a data-dependent manner (e.g., via cross-validation), introducing an additional layer of selection that is typically not accounted for.

To formally illustrate the above example, let $\widehat\lambda(D, W^{(1)})$ denote the selection procedure for the regularization parameter, and let $\widehat M^\lambda(D, W^{(2)})$ denote the lasso variable selection procedure with regularization parameter $\lambda$. 
Both procedures may involve randomization, represented by the randomization variables $W^{(1)}$ and $W^{(2)}$, respectively.
Suppose that for the observed data, the selected regularization parameter is $\ob{\lambda}=\widehat\lambda(\ob{D}, \ob{W}^{(1)} )$, and the subsequent lasso with regularization parameter $\ob\lambda$ selects the model $\ob{M}=\widehat M^{\ob \lambda}(\ob{D}, \ob{W}^{(2)} ) $. This results in the selected linear model $y\sim \N(X_{\ob M}\beta_{\ob M}, \sigma^2 I_n)$, where $X_{\ob M}$ is the design matrix composed of covariates indexed by $\ob M$. The test statistic is chosen to be the least squares estimator in the selected model $T=\widehat\beta_{\ob M}=(X_{\ob M}\tran X_{\ob M})^{-1}X_{\ob M}\tran y$. 
In line with existing methods for the lasso, we consider conditioning on the additional statistic $A = (I_n - \calP_{X_{\ob{M}}}) y$, where $\calP_{X_{\ob{M}}}$ denotes the projection onto the column span of $X_{\ob{M}}$, as well as on the sign vector $S$ of the lasso solution.
Conditioning on this extra information enables efficient computation of the lasso selection probability at a fixed value of the tuning parameter.

To perform valid inference for $\thetaobs$, the post-selection parameter of interest after selecting $\ob{M}$, we must now account for both steps of selection. This parameter is formally defined in \eqref{eqn: ps-target} in the following section.
The result below characterizes the conditional density of $T$, given the selected regularization parameter, the active variable set, and the associated additional statistics.
\begin{proposition}\label{prop: lasso}
The density of the conditional distribution of $T$  given \sloppy{$\{\widehat\lambda = \ob{\lambda}, \widehat M^{\ob{\lambda}} = \ob{M},A=\ob{A}, S=\ob{S} \}$}, evaluated at $t$, is given by
\begin{equation}\label{equ: lasso density factor}
p^*_{\theta^{\ob{M}}}(t) \propto  \underbrace{p_{\thetaobs}(t\mid \widehat\lambda=\lambda, A=\ob{A} ) }_{(I)} \cdot \underbrace{\PP{\widehat M^{\ob{\lambda}} =\ob{M}, S=\ob{S}\mid T=t, A=\ob{A}}}_{(II)},
\end{equation}
where $(\Rom{1})$ is the conditional density of $T$ given $\{\widehat\lambda=\lambda,A=\ob A\}$, and $(\Rom 2)$ is the lasso selection probability conditioned on both $T$ and $A$ at regularization parameter $\ob\lambda$.
\label{prop: cond:den:MQ}
\end{proposition}
\begin{proof}[Proof of Proposition~\ref{prop: lasso}]
We can write the conditional density of $T$ as
\begin{align*}
    p_{\thetaobs}^*(t) &= p_{\thetaobs}(t \mid\widehat\lambda=\ob\lambda, \widehat{M}^{\ob\lambda} = \ob M, A=\ob A, S=\ob{S} )\\
    &\propto p_{\thetaobs}(t\mid \widehat\lambda=\ob\lambda, A=\ob A )\cdot \PP{\widehat{M}^{\ob\lambda}=\ob M, S=\ob{S} \mid  T=t, A=\ob A, \widehat\lambda=\ob\lambda }.
\end{align*}
Because the lasso selection $\widehat{M}^{\lambda_0}$ is conditionally independent of $\widehat\lambda$ when conditioned on $T$ and $A$---lasso selection is completely characterized by $T$ and $ (I_n - \calP_{X_{\ob M}} )y$---we have
\begin{align*}
    \PP{\widehat{M}^{\ob\lambda}=\ob M, S=\ob{S} \mid  T=t, A=\ob A, \hat\lambda=\ob\lambda } = \PP{\widehat{M}^{\ob\lambda}=\ob M, S=\ob{S} \mid  T=t, A=\ob A}.
\end{align*}
The probability is over the possible randomness in the external randomness $W^{(2)}$, and does not depend on the model parameter $\thetaobs$. This finishes the proof.
\end{proof}

The second term $(\Rom{2})$ in the conditional density~\eqref{equ: lasso density factor} is the probability of the lasso selection (including the sign) with a fixed regularization parameter. The lasso selection event $\{\widehat M^{\ob\lambda}=\ob{M}, S=\ob{S} \}$ is known to be a polyhedron if no external randomization is applied. In this case, the selection probability is an indicator function and can be computed exactly \citep{lee2016exact}. When external randomization $W^{(2)}$ is introduced, the selection probability can be expressed as the probability of a Gaussian distribution over a polyhedral region.
Efficient Monte Carlo estimators for this probability have been developed in \cite{liu2023exact}, or the exact probability can be computed with some more additional conditioning in \cite{panigrahi2024exact}.

On the other hand, the first term $(\Rom 1)$, which is the conditional density of $T\mid \widehat\lambda=\ob\lambda, A=\ob A$, is harder to characterize, particularly when $\widehat\lambda$ is selected via some complex tuning procedure, such as cross-validation. This is precisely where our proposed method proves especially useful, offering a practical tool to account for the adaptive selection of the regularization parameter.
Specifically, we apply the proposed method to learn a transport map that pushes forward the pre-selection distribution $\bbP_{\thetaobs}$ to the conditional distribution of $T\mid \hat\lambda=\ob\lambda, A=\ob A$. The transport map provides an approximation of the conditional density $(\Rom 1)$, analogous to the expression given in Equation~\eqref{push:forward:density}. By combining the approximate conditional density $(\Rom 1)$ with the computed lasso selection probability $(\Rom 2)$, we obtain a valid selective inference procedure that accounts for both steps of selection. 

We present this example in Section~\ref{sec: lassocv} and demonstrate how our method, combined with existing approaches, delivers valid selective inference on the full dataset. For now, we illustrate the conditional density $(\Rom 1)$ and the lasso selection probability $(\Rom 2)$ in Figure~\ref{fig: random lassocv demo}, highlighting how our method successfully corrects for the complex process of tuning parameter selection.
The left panel plots the conditional density of $(\Rom 1)$ (blue solid line), computed using the learned transport map, alongside the pre-selection density of $\hat\beta_1$ (orange dashed line), where $\hat\beta_1$ denotes the least squares estimate for a variable included in the selected model. The middle panel shows the lasso selection probability $(\Rom 2)$ as a function of $\hat\beta_1$, computed using the Monte Carlo method of \cite{liu2023exact}. In the right panel, the blue solid line represents the product of $(\Rom 1)$ and $(\Rom 2)$---that is, the estimated conditional density accounting for both stages of selection---while the orange dashed line represents the pre-selection density multiplied by $(\Rom 2)$, corresponding to inference that ignores the data-adaptive nature of $\lambda$. 

The red dotted vertical line in the plots indicates the observed value of $\hat{\beta}_1$. 
Under the orange curve in the right panel, the observed value lies to the right of the 0.99 quantile, leading to a false rejection of the null hypothesis $\beta_1 = 0$. In contrast, under the blue curve—which incorporates the selection of $\lambda$—the observed value falls within the 0.95 quantile, and the null is not rejected. This example illustrates that failing to account for the selection of the regularization parameter leads to invalid inference for the regression coefficients in the selected model.

\begin{figure}
  \centering
  \includegraphics[width=.9\textwidth]{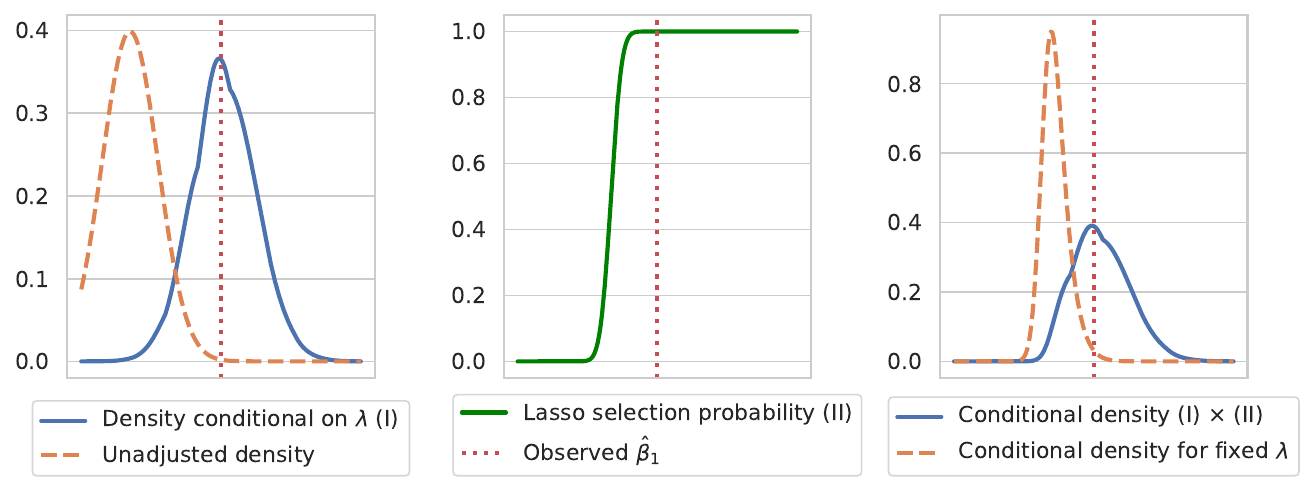}
  \caption{Left: pre-selection density (orange dashed) and conditional density $(\Rom 1)$ of $\hat\beta_1 \mid \hat\lambda = \ob{\lambda}, A = \ob{A}$ (blue solid), obtained via the learned transport map. Middle: lasso selection probability $(\Rom 2)$, estimated using \citet{liu2023exact}. Right: conditional density ignoring $\lambda$ selection (orange dashed) vs. full conditional density accounting for both selection steps (blue solid). The red vertical line marks the observed value of $\hat\beta_1$.}
  \label{fig: random lassocv demo}
\end{figure}

\section{Applications}
\label{sec: experiments}

In this section, we illustrate the performance of our proposed method on several simulation experiments and a single-cell data analysis. In all the examples below, we fit a generalized linear model (GLM) to the data, using a feature matrix $Z_{\ob{M}}\in\R^{n\times d}$ that is constructed adaptively based on the observed data, treating the response $y\in\R^n$ as random and the predictor matrix $X$ as fixed. 
We consider four different selection procedures for constructing the feature matrix:
\begin{enumerate}
    \item In Section~\ref{sec: polynomial experiment}, the feature matrix is constructed using a polynomial basis expansion, where the degree of the polynomial fit is selected based on an analysis of variance (ANOVA) criterion.
    \item In Section~\ref{sec: spline experiment}, the feature matrix consists of natural cubic splines, where the number of knots in the nonlinear fit is selected via cross-validation (CV).
    \item In Section~\ref{sec: lassocv}, the feature matrix is composed of a subset of predictors selected by the lasso, with the regularization parameter first chosen via CV.
    \item In Section~\ref{sec: pcr experiment}, the feature matrix consists of the principal components of the predictor matrix, following the standard approach in principal component regression (PCR), with the number of principal components chosen via CV. 
\end{enumerate}

Given the adaptively constructed feature matrix $Z_{\ob{M}}$, the selected GLM has density given by
\begin{align*}
    p_{\ob M}(y\mid Z_{\ob M}, \beta_{\ob M}) = \exp \Big(\frac{y\tran Z_{\ob M}\beta_{\ob M} - A(\beta_{\ob M})}{\sigma^2} \Big)\cdot h(y;\sigma),
\end{align*}
where $\beta_{\ob M}\in\R^d$ is the vector of regression coefficients. We assume that the dispersion parameter $\sigma^2$ is either known or can be consistently estimated from the data, allowing us to use this estimate as a plug-in value in our inference procedure. The examples in  Sections~\ref{sec: polynomial experiment}, \ref{sec: spline experiment}, and \ref{sec: lassocv} involve a Gaussian linear model, and the example in Section~\ref{sec: pcr experiment}  considers a logistic regression model.

Following \citep{berk2013valid, lee2016exact}, the target post-selection parameter in our simulations is defined as the best linear regression coefficients in the selected model, and is given by
\begin{align}
    \thetaobs = \underset{\beta\in\R^{d}} {\argmin} \; \EE{\sum_{i=1}^n \ell(\beta;y_i, Z_{\ob M, i} ) },
    \label{eqn: ps-target}
\end{align}
where $\ell(\beta; y, Z) = A(\beta) - y Z^\top \beta$ is the negative log-likelihood of the selected linear model. 
The expectation in \eqref{eqn: ps-target} is taken with respect to the true data-generating distribution.
Note that the selected model $\ob M$ is treated as a working model, and we make no assumptions about the quality of the selection, allowing $\ob M$ to be potentially misspecified.

In the proposed method, when constructing confidence intervals for individual parameters, we train a conditional normalizing flow as described in Section~\ref{sec: conditional NF}. If the goal is to test a specific null hypothesis (as in Section~\ref{sec: spline experiment}), we instead train an unconditional normalizing flow as outlined in Section~\ref{sec: NF}.
The normalizing flows are parametrized using RealNVP~\citep{dinh2017density} with 12 affine coupling layers. Each coupling layer involves a shift and a scaling parameter, both of which are outputs of neural networks with 1 hidden layer containing 8 neurons. For conditional normalizing flows, these neural networks take the parameter value $\theta$ as an additional input.
For each simulation, we generate 2,000 training samples and 500 validation samples. The normalizing flows are trained using full-batch Adam with learning rate $10^{-4}$ for 10000 iterations. 
The validation loss---KL divergence computed on the validation set---is computed every 1000 iterations, and the flow parameter corresponding to the lowest validation loss is selected as the final parameter. More details about the normalizing flow and the training procedure are provided in Appendix~\ref{sec: experiment details}. Code for reproducing the experiments is available at \url{https://github.com/liusf15/transport_selinf}.

\subsection{Selecting the degree of polynomial fit}
\label{sec: polynomial experiment}

When fitting a polynomial regression model with $y$ as the response and $x$ as the predictor, the degree of the polynomial is typically unknown a priori.
To determine this unknown degree, we apply an ANOVA procedure, in which we start from a polynomial of degree 0, and sequentially increase the degree in a stepwise manner. Suppose $\mathbb{M}_p$ is the model with polynomial degree $p$. An F-test is used to compute the p-value for comparing the model $\mathbb{M}_p$ with the model $\mathbb{M}_{p+1}$, treating the smaller model as the null hypothesis. If the p-value is larger than 0.05, we stop and choose $\mathbb{M}_p$ as the final model; otherwise, we continue the procedure to compare $\mathbb{M}_{p+1}$ and $\mathbb{M}_{p+2}$. The procedure is stopped when the maximum degree is reached.
For a detailed description of this procedure, we refer the readers to \cite[Chapter 7]{james2013introduction}.

When a degree $p > 0$ is selected, we call the selected model $\ob{M}$ and construct the corresponding $n \times (p + 1)$ feature matrix $Z_{\ob{M}}$, where the $k$-th column contains $x^k$ for $0 \leq k \leq p$. We then fit a least squares regression using this feature matrix and construct confidence intervals for each regression coefficient in the model.
In this simulation, we set $n=100$ and generate $x_i\iid\N(0,1)$, $y_i\sim \N(f^*(x_i), 1)$ ($1\leq i\leq n$), where $f^*(x_i) =c\cdot (x_i^3 + x_i^4)$ with the signal strength $c$ varied over $\{0, 0.1, 0.2\}$. We refer to these as signal regimes SR-$1$, SR-$2$, SR-$3$, respectively. We set the maximum degree of the polynomial fit to be 4. 

Our method enables valid inference in the selected model by explicitly accounting for the selection of the polynomial degree $p$ from the ANOVA procedure. 
We compare the proposed method to the following approaches:
\begin{enumerate}
\item[(i)] \textbf{Na{\"i}ve}: which performs inference based on the least squares estimator from the selected model, but is clearly invalid as it ignores the data-adaptive choice of $p$. 
\item[(ii)] \textbf{Splitting}: which uses the UV method by \cite{rasines2021splitting}. The UV method performs selection on a perturbed version of the response, in this case $y + W$, where $W \sim \mathcal{N}(0, \nu^2 I_n)$ for a pre-specified value of $\nu^2$. It then conducts inference on $y-\frac{\sigma^2}{\nu^2}W$, the holdout portion of the response, which is independent of the data used for selection. 
This approach is a variant of data splitting for fixed $X$ regression, and falls under the broader class of data fission or thinning schemes described in \cite{leiner2025data, dharamshi2025generalized}. 
\end{enumerate}
\begin{figure}
    \centering
    \includegraphics[width=\textwidth]{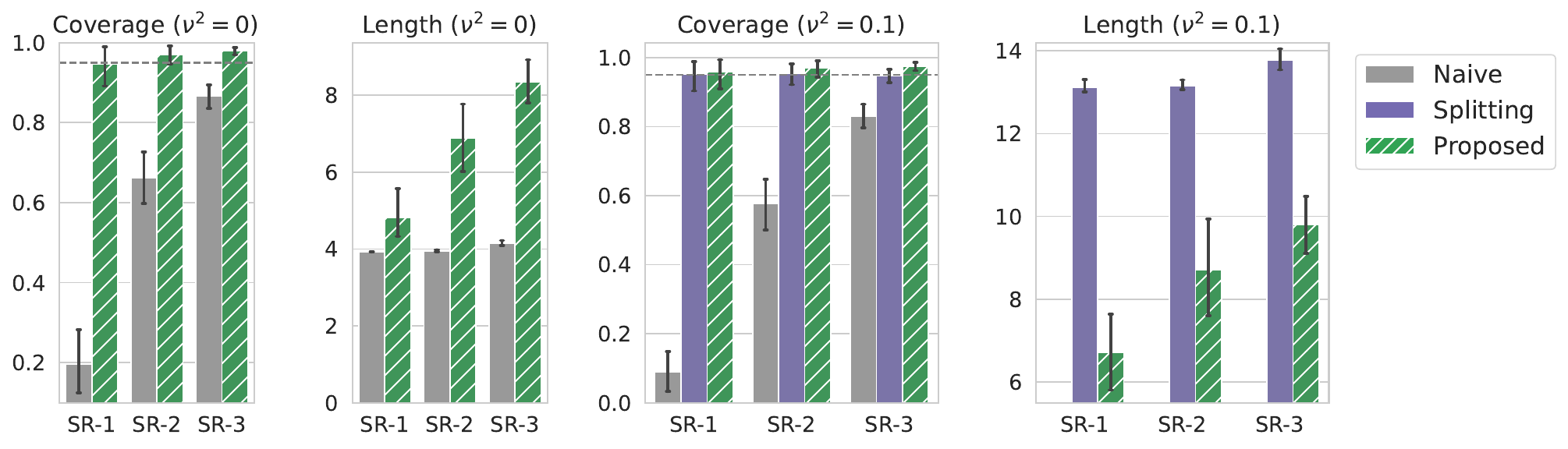}
    \caption{Coverage proportions and average interval lengths for the coefficients in polynomial regression, where the polynomial degree is selected by an ANOVA criterion. The first two plots show the coverage and interval lengths for non-randomized selection with $\nu^2=0$, and the last two plots are for $\nu^2=0.1$. The dashed line indicates the target coverage probability 0.95.
    }
    \label{fig: polynomial}
\end{figure}

We consider two settings for our simulations: $\nu^2 = 0$ (no randomization) and $\nu^2 = 0.1$ (a small amount of randomization).
In the latter case, the polynomial fit is expected to be very close to that of the non-randomized procedure based on the full data. While our approach applies to both randomized and non-randomized selection, the splitting approach is applicable only to the latter case.

Figure~\ref{fig: polynomial} shows the coverage probabilities and interval lengths based on 2000 simulation repetitions for each signal regime. A non-zero degree is selected in approximately 100 simulations for SR-1, and in about 400 simulations for SR-3.
As can be seen from the empirical coverage probabilities in this plot, na{\"i}ve inference  leads to severe under-coverage, whether or not randomization is used in the selection procedure. 
In contrast, our method produces valid confidence intervals that achieve the nominal coverage probability across all scenarios.
In the case where $\nu^2 = 0.1$, our method, following the principles of data carving, uses the full data for inference and, as a result, produces significantly shorter confidence intervals for the post-selection parameters. In particular, the splitting intervals are nearly $1.5-2$ times longer than those produced by our data-carving method.

\subsection{Selecting the number of knots}
\label{sec: spline experiment}

This example extends the analysis from Section \ref{sec: setup}, where we fit regression splines to the data.
The number of knots $K$ in the fit is a hyperparameter, which is selected using a CV procedure. 
As described in Section \ref{sec: setup}, we test the null hypothesis that there is no association between $y$ and $x$. However, since $K$ is selected adaptively based on the data, the classical F-test for this task fail to control the Type~$\Rom{1}$ error.

We select the number of knots from $\{2,3,4,5\}$ using a 10-fold CV procedure. We denote by $\ob{M}$ the model corresponding to the selected $K$, which uses the feature matrix $Z_{\ob{M}}$ consisting of $(K + 1)$ basis functions along with an intercept. We then perform inference using the nonlinear model fitted with this adaptively constructed feature matrix. For our simulations, we fix $n=100$, and generate the predictors $x_i\iid \unif(0,1)$ for $1\leq i\leq n$. The response variables are then generated as $y_i\sim \N(f^*(x_i), 1)$ where $f^*(x)$ is the natural cubic spline function with 3 knots positioned at the quartiles. The coefficients for the 4 basis functions are given by $c \cdot (1, 1, -1, 1)$, with the signal strength $c$ varying over the set $\{0, 0.1, 0.2, 0.3\}$. This leads to four signal regimes, which we denote by SR-$0$, SR-$1$, SR-$2$, and SR-$3$, respectively.
When $c=0$, $f^*$ does not depend on $x$ and the null hypothesis is true. When $c>0$, the alternative hypothesis is true.

To test the null hypothesis that $y$ has no significant association with $x$ using our data-carving method, we follow the pipeline outlined in Figure~\ref{fig: pipeline}. In this case, we choose the standard normal distribution as the reference distribution. Specifically, we obtain a transport map $\hat\tau$ such that $\hat\tau\#\N(0, I_d)$ approximates $\bbP^*_{0}$, the conditional distribution of the least squares estimator $T$ under the null. We then compute the pulled-back statistic $\hat\tau^{-1}(T)$, using the inverse transport map, and reject the null hypothesis if $\|\hat\tau^{-1}(T)\|_2^2 > \chi^2_{(K+1), 1-\alpha}$, where $\chi^2_{(K+1), 1-\alpha}$ is the $1-\alpha$ quantile of the chi-squared distribution with $K+1$ degrees of freedom. 

In line with the previous example, we conduct our simulations under two types of selection procedures: one without additional randomization and one with randomization, corresponding to $\nu^2 = 0$ and $\nu^2 = 0.1$, respectively. We compare our method to the invalid na{\"i}ve approach, which performs a classical F-test to compare the intercept-only model with the selected linear model, without accounting for the prior selection of $K$. When $\nu^2 = 0.1$, we also compare our method to the data-splitting approach, which applies na{\"i}ve  inference to the holdout portion of the response, as described in Section~\ref{sec: polynomial experiment}.

 Figure~\ref{fig: spline} presents the empirical selective Type~$\Rom{1}$ error and the power of the tests from 500 simulation runs. As expected, our method offers valid control on the Type $\Rom 1$ error, and achieves substantially higher power than the splitting method.

\begin{figure}
    \centering
    \includegraphics[width=\textwidth]{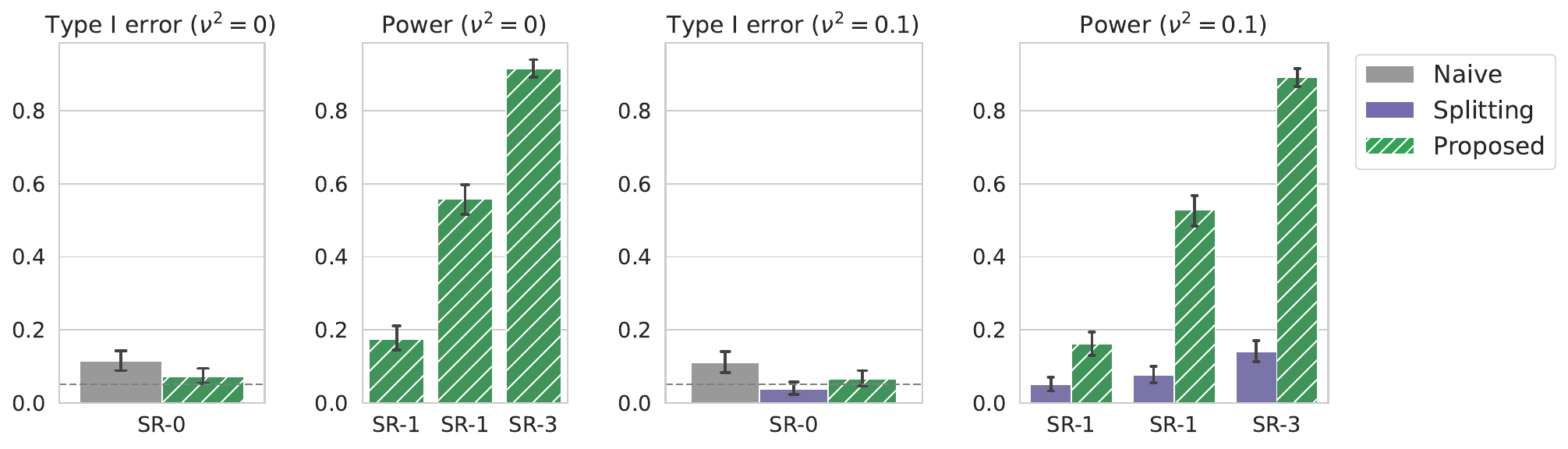}
    \caption{Results for testing the dependence between $y$ and $x$ by fitting a spline regression model, where the number of knots is chosen by CV. The first two plots show the Type $\Rom 1$ error and power for the non-randomized selection with $\nu^2=0$, and the last two plots are for $\nu^2=0.1$. The dashed line indicates the nominal Type $\Rom 1$ error 0.05. }
    \label{fig: spline}
\end{figure}

\subsection{Lasso with cross-validated regularization parameter}
\label{sec: lassocv}

We consider the example of the lasso with a cross-validated regularization parameter, following the discussion in Section~\ref{sec: multiple queries}.
Given a regularization parameter $\lambda>0$, the lasso solves the optimization problem
\begin{align*}
    \hat\beta^\lambda = \underset{\beta\in\R^p}{\argmin}\; \frac12\|y - X\beta\|_2^2 + \lambda\|\beta\|_1,
\end{align*}
where $y\in\R^n$ and $X\in\R^{n\times p}$ denote the response vector and the predictor matrix, respectively. The $\ell_1$-penalty induces sparse solutions, leading to a selected set of variables:
\begin{align*}
    \widehat M^\lambda(y, X) = \{j\in[p]: \hat\beta_j^\lambda\neq 0\}.
\end{align*}
Given that we observe $\ob M= \widehat M^\lambda(y, X)$, we perform inference in the linear model $y\sim \N(Z_{\ob M}\beta_{\ob M}, \sigma^2 I_n)$, where $Z_{\ob M} = X_{\ob M}$ is the submatrix of $X$ corresponding to the variables in the subset $\ob M$.

Inference after the lasso at a fixed regularization parameter has been extensively studied in the selective inference literature, with many approaches proposed to tackle this problem. If one conditions additionally on the sign of $\hat\beta^\lambda$, then \cite{lee2016exact} showed that the selection event can be characterized as a polyhedron, and in this case, selective inference can be based on the distribution of a univariate truncated normal variable. But this method often leads to excessively conservative confidence intervals, as shown by \cite{kivaranovic2021length}. 

To increase power and avoid infinitely long confidence intervals, the randomized lasso of the form
\begin{align}\label{equ: random lasso}
    \hat\beta^\lambda =\argmin_{\beta\in\R^p} \frac12\|y - X\beta\|_2^2 + \lambda\|\beta\|_1 - W\tran\beta,
\end{align}
was used in \cite{panigrahi2022approximate, liu2023exact, panigrahi2024exact}, where $W\sim \N(0, \nu^2 X\tran X)$ is a $p$-dimensional randomization variable, and $\nu^2$ controls the randomization level. 
Conditional on the sign of $\hat\beta^\lambda$, the distribution of the least squares estimator in the selected model is a soft-truncated normal distribution, after marginalizing over all or part of the randomness in $W$.
Here, we consider two selective inference methods after solving \eqref{equ: random lasso} at fixed $\lambda$: (i) the separation-of-variable (SOV) method by \cite{liu2023exact}, which computes the quantiles of the soft-truncated distribution with Monte Carlo sampling; (ii) the bivariate normal (Bivnormal) method by \cite{panigrahi2024exact}, which computes an exact pivot by conditioning on additional information.

A closely related randomization scheme involves solving the lasso using a randomized response $y + W$, where $W \sim \mathcal{N}(0, \nu^2 I_n)$. This is introduced by \cite{tian2018selective} and is equivalent to solving the randomized lasso in \eqref{equ: random lasso}.
Using the variables selected with the noisy response vector, the UV method of \cite{rasines2021splitting} can be used to obtain splitting-type inference based on $y - \dfrac{\sigma^2}{\nu^2} W$.

In our simulations, we set $n = 100$ and $p = 20$. The predictor matrix $X\in\R^{n\times p}$ is generated as $X_i\iid\N(0, \Sigma_X)$ for $1\leq i\leq n$, where $\Sigma_{X,jk}=0.9^{|j-k|}$, and the response vector $y$ is generated from $N(0,I_n)$, an all-noise model.
We perform a 10-fold CV to select $\lambda$ over a pre-specified grid ranging from $0.01\sqrt{{\log p}/{n}} $ to $5\sqrt{{\log p}/{n}} $, equally spaced on the log scale. 
We consider both non-randomized ($\nu^2=0$) selection and selection with a small amount of randomization ($\nu^2=0.1$). When $\nu^2=0.1$, both cross-validation and lasso selection are performed using the perturbed data $y+W$, where $W\sim\N(0, \nu^2 I_n)$.

Obviously, the na{\"i}ve approach which assumes $\hat\beta \sim \N(\beta_{\ob M}, \Sigma)$ and ignores all selection steps preceding inference, is clearly invalid.
Furthermore, we consider the following methods that treat $\ob\lambda$ as fixed and offer only partial corrections for the variable selection process:
\begin{enumerate}
    \item \textbf{Polyhedral}: inference is based on the truncated normal distribution of $\hat\beta$ conditioned on the lasso selection event \citep{lee2016exact}. 
    \item \textbf{Bivnormal}: quantiles of $\hat\beta_j$ are computed based on a bivariate normal distribution \citep{panigrahi2024exact}, which marginalizes over an appropriately-chosen linear projection of $W$.
    \item \textbf{SOV}: quantiles of $\hat\beta_j$ are computed using the sampling method from \cite{liu2023exact}, which marginalizes over all the randomness in $W$.
\end{enumerate}

\begin{figure}
    \centering
    \includegraphics[width=\textwidth]{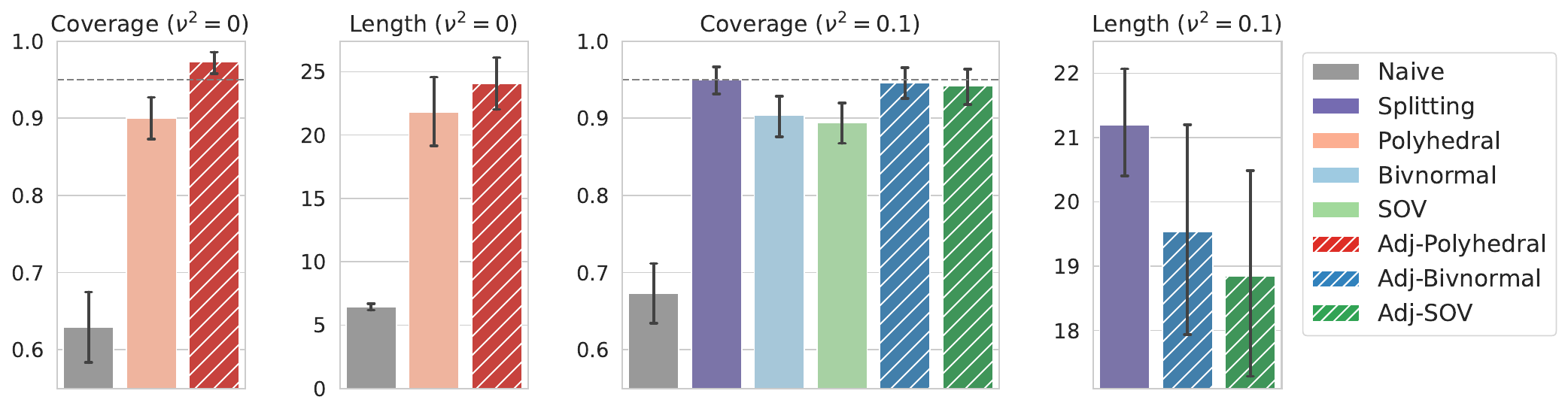}
    \caption{Lasso with cross-validation. The first two panels show the coverage proportions and interval lengths for the scenario $\nu^2=0$, while the last two panels are for $\nu^2=0.1$. The horizontal dashed line indicates the target coverage probability 0.95. }
    \label{fig: lasso results}
\end{figure}

These methods are partially invalid because they ignore the selection of $\ob \lambda$. We can adjust for the selection of $\ob\lambda$ by combining these methods with our proposed method, as described in Section~\ref{sec: multiple queries}.
These adjusted methods are referred to as \textbf{Adj-Polyhedral}, \textbf{Adj-Bivnormal}, \textbf{Adj-SOV}.
To implement \textbf{Adj-Polyhedral}, we multiply the conditional density for the $j^{\text{th}}$ regression coefficient---accounting for the selection of $\lambda$ via our method, that corresponds to the first term $(\Rom 1)$ in Proposition~\ref{prop: lasso}---with $1_{[I^j_1,I^j_2]}(t)$, where  $[I^j_1, I^j_2]$ denotes the truncation interval obtained with \textbf{Polyhedral}. For \textbf{Adj-Bivnormal}, we multiply this conditional density, obtained using the learned transport map, with $H(t)=\Phi(a^jt + b^j_2)- \Phi(a^jt+ b^j_1)$, where $a^j, b^j_1, b^j_2$ are scalars obtained with \textbf{Bivnormal}. Similarly, \textbf{Adj-SOV} is obtained by multiplying this conditional density, based on the learned transport map, with an estimate of the lasso selection probability using Monte Carlo sampling obtained with \textbf{SOV}.
In the case when $\nu^2 = 0.1$, we include the splitting-based UV method, which performs inference using simply the holdout data $y-\frac{\sigma^2}{\nu^2} W $. 

The simulation is repeated 1000 times, out of which about 270 times the selected model is nonempty. If the selected model is nonempty, we compute the average coverage proportions of all the selected coefficients and the average interval lengths. The results are reported in Figure~\ref{fig: lasso results}.
Methods that do not account for the selection of $\ob\lambda$ consistently exhibit under-coverage. However, when combined with our proposed approach, all these methods---both with and without additional randomization---are able to achieve the target 95\% coverage probability. As expected, the valid methods produce wider intervals, taking into account the additional uncertainty from the selection of the regularization parameter. 
Consistent with findings from the literature on randomized selective inference, the data-carving intervals from \textbf{Adj-Bivnormal} and \textbf{Adj-SOV}, obtained using the randomized lasso, are shorter than those from the other valid approaches, including \textbf{Splitting} and \textbf{Adj-Polyhedral}.

\subsection{Selecting the number of principal components}
\label{sec: pcr experiment}

Below, we consider performing a principal component regression (PCR) analysis with a Bernoulli random variable, where the number of principal components (PCs) is first selected via CV. 
Let $K$ denote the number of PCs, which we select using a 5-fold CV on the data. Let $Z_{\ob{M}} \in \R^{n\times K}$ denote the feature matrix consisting of the top $K$ PCs of $X$. We then fit the logistic regression model
\begin{align*}
    y_i\mid Z_i \sim \mathrm{Bernoulli}( (1 + \exp(-\beta_0 -Z_{\ob{M},i}\tran \beta))^{-1} ),
\end{align*}
where $\beta_0$ represents the intercept.
We begin with a simulation study, followed by an application in single-cell data analysis.

\paragraph{Simulations} For our simulations, the predictors are generated as $x_i \iid \N_p(0, \Sigma)$, where the covariance matrix $\Sigma$ has entries $\Sigma_{ij} = \rho^{|i-j|}$. The responses are generated as $y_i \iid \mathrm{Bernoulli}(1/2)$, with no dependence between $y$ and $x$. We fix $n = 100$, $p = 50$, and vary $\rho \in \{0.3, 0.6, 0.9\}$, which are labeled as Corr-$1$, Corr-$2$, and Corr-$3$, respectively, in the plots.
Inference is only performed when the selected $K>0$. After selection, we construct confidence intervals for the regression coefficients, and test the global null hypothesis that all $\beta_j = 0$ for $1 \leq j \leq K$.

Ignoring the selection step, one would construct Wald-type confidence intervals for the individual regression coefficients $\beta_j$ for $1 \leq j \leq K$, and perform a likelihood ratio test (LRT) for the global null hypothesis. However, due to selection bias, the resulting confidence intervals would fail to achieve the nominal coverage probabilities, and the LRT would not properly control the Type $\Rom{1}$ error. Our method, by contrast, accounts for the selection of the number of principal components, providing valid tests and confidence intervals for the post-selection parameter.

The simulation is repeated 1000 times for each scenario, among which about 300-400 simulations select a non-zero $K$.
The results are shown in Figure~\ref{fig: pcr}.
The confidence intervals constructed by our data-carving method achieve the target coverage probability, with lengths only slightly longer than those of the na{\"i}ve Wald intervals. Additionally, our method provides a valid test for the global null hypothesis, with proper control of the selective Type~$\Rom{1}$ error as expected.

\begin{figure}
    \centering
    \includegraphics[width=\textwidth]{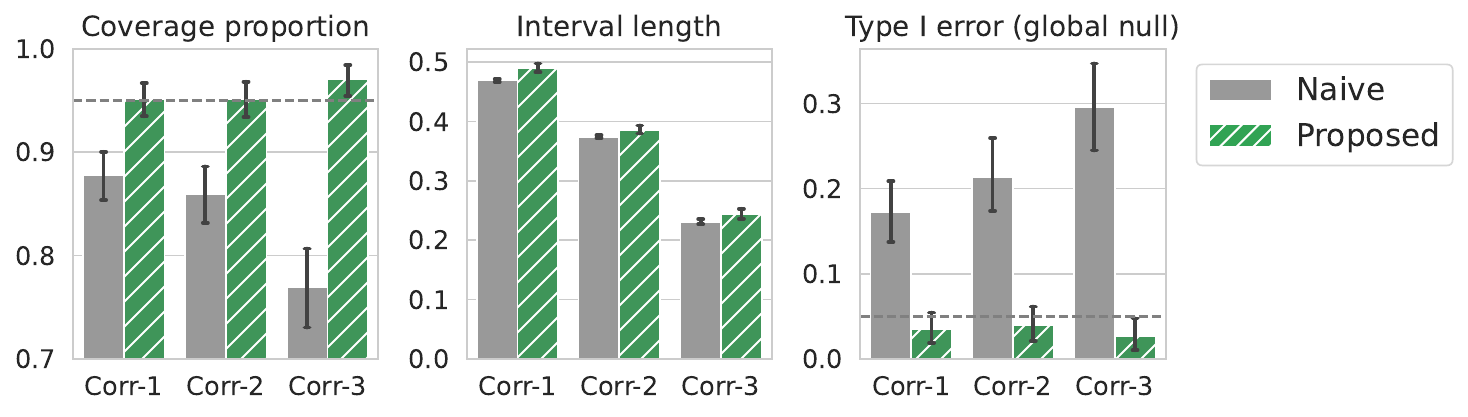}
    \caption{Results for the principal component regression example. Left panel: average coverage proportions of regression coefficients $\beta_j$ in the selected PCR model. Middle panel: average interval lengths for the confidence intervals of $\beta_j$. Right panel: Type $\Rom{1}$ error for testing the global null. }
    \label{fig: pcr}
\end{figure}

\paragraph{Application in single-cell gene expression analysis} Single-cell RNA sequencing (scRNA-seq) enables researchers to profile gene expression at the resolution of individual cells, allowing for the discovery and characterization of highly specialized cell types \citep{regev2017human}. Due to the high dimensionality of gene expression data, principal component analysis (PCA) is widely used for dimensionality reduction, and is applied as a standard preprocessing step in popular tools for analyzing scRNA-seq data, such as the \texttt{Seurat} R package \citep{satija2015spatial, stuart2019comprehensive, hao2021integrated}.
One common application involves using PCA to derive features before fitting a model in a supervised analysis.
For example, the prediction algorithm \texttt{scPred} \citep{alquicira2019scpred} trains classifiers using the first few PCs to perform tasks such as predicting cell types in peripheral blood mononuclear cells (PBMCs). The number of PCs to retain is typically selected using heuristics like the elbow plot. More principled methods, such as data thinning or cross-validation, can also be applied for selection, as demonstrated in \cite{neufeld2024data}.

Following a similar approach to \citep{alquicira2019scpred}, we focus on distinguishing memory B cells from na\"ive B cells in a PBMC dataset. We use the publicly available 10X Genomics dataset\footnote{\url{https://www.10xgenomics.com/datasets/5k_Human_Donor1_PBMC_3p_gem-x}} and fit a logistic PCR model, as described earlier in this section. The number of PCs included in the logistic regression fit is selected using a five-fold CV. Preprocessing steps follow the guidelines in the \texttt{Seurat} tutorial\footnote{\url{https://satijalab.org/seurat/articles/pbmc3k_tutorial}}, with details provided in Appendix~\ref{sec: preprocessing details}. After preprocessing, we retain 2000 genes and 233 cells annotated as either memory B cells (140 cells) or na\"ve B cells (93 cells). 

CV selects the top six PCs, which are then used to construct features for the logistic regression model. Our focus is on constructing p-values and confidence intervals for the coefficients in the fitted logistic regression model, as commonly reported in inference summary tables.
As discussed earlier, the na{\"i}ve inferential approach that ignores the presence of the CV procedure used to construct the feature matrix will likely produce overly optimistic and misleading conclusions. In contrast, our proposed method offers a principled approach to valid inference by accounting for the selection procedure while utilizing the entire dataset for this task.

Table~\ref{table: pbmc} summarizes our results. The na{\"i}ve method reports the first three PCs as statistically significant, whereas our method identifies only the first two. For the remaining coefficients, our method produces wider confidence intervals than the na\"ive method, reflecting the additional uncertainty introduced by the data-driven construction of features.



\begin{table}[h]
\centering
\begin{tabular}{c|cc|cc}
\toprule
 & \multicolumn{2}{c|}{Na\"ive} & \multicolumn{2}{c}{Proposed} \\
PC & p-value & CI & p-value & CI \\
\midrule
1 & 0.000 & {$(-1.913, -0.581)$} & 0.005 & $(-1.840, -0.359)$ \\
2 & 0.000 & {$(-4.117, -2.375)$} & 0.001 & $(-4.590, -2.676)$ \\
3 & 0.022 & {$(-1.609, -0.125)$} & 0.783 & $(-2.947, \phantom{-}1.170)$ \\
4 & 0.516 & {$(-0.810,\phantom{-} 0.407)$} & 0.834 & $(-1.932, \phantom{-}1.106)$ \\
5 & 0.588 & {$(-0.379,\phantom{-} 0.668)$} & 0.264 & $(-0.353, \phantom{-}1.175)$ \\
6 & 0.136 & {$(-0.162,\phantom{-} 1.186)$} & 0.889 & $(-4.660, \phantom{-}1.323)$ \\
\bottomrule
\end{tabular}
\caption{P-values and 95\% confidence intervals for the regression coefficients from the selected logistic PCR in the PBMC data analysis. }
\label{table: pbmc}
\end{table}

\section{Concluding remarks}
\label{sec:conc}

In this paper, we introduce a data-carving method that enables powerful and flexible selective inference with conditional guarantees. On the one hand, unlike approaches such as data-splitting—that condition on the data used for selection, equivalent to conditioning on much more than necessary—our method reuses data from the selection steps by conditioning on less. On the other hand, unlike existing data-carving methods—which rely on an explicit analytical characterization of the selection event—our method can handle selection events for which no such description is available. 

Our data-carving method applies to selection procedures both with and without additional randomization.  In cases where certain types of selection events are known to suffer a loss in power without external randomization, our method, similar to approaches in randomized selective inference, can take into account the randomized selection procedure to improve power. However, in contrast to these existing methods, which rely on a specific form of randomization to make selective inference feasible, our approach in this paper does not rely on the form of the randomization used. For example, in our simulations, we used a standard CV procedure based on sample splitting to select tuning parameters. However, we note that our method can be readily applied to alternative CV techniques, such as the antithetic CV proposed in \cite{liu2024cross}, which employs a correlated Gaussian randomization scheme to select tuning parameters.

The key statistical idea underlying our approach is simple: we use a pushforward transport map from a simple reference density to the conditional distribution, and then apply the inverse map---the pullback map---to perform selective inferential tasks such as hypothesis testing and interval estimation. To efficiently learn the pullback map, we employ a normalizing flow. More broadly, our work demonstrates how powerful tools from generative modeling can be utilized to broaden the scope of selective inference methodology, while still ensuring strong conditional guarantees, such as control of the selective Type $\Rom{1}$ error and selective coverage probability.

Finally, we identify several promising directions for future research. For example, in the spirit of simulation-based inference, other types of parameterizations of the transport map, such as the generative neural networks in \cite{liu2021density}, could be used in place of normalizing flows for learning the target conditional distribution. Potential extensions of our work include developing selective inference within a Bayesian framework, which involves sampling from a posterior formed by appending a prior to the selective likelihood. However, similar to the frequentist line of work, existing methods have primarily focused on selection events with analytical characterizations. By contrast, new extensions could enable selective inference for potentially intractable posteriors, allowing data-scientists to leverage the full benefits of the Bayesian framework. 

\appendix

\section{Derivation of the conditional density}
\label{app: conditional density}

\begin{proposition}
    The conditional density of $T\mid\widehat M=\ob M$ is proportional to
    \begin{align*}
         p_{\theta^{\ob{M}}}(t) \times \PP[\thetaobs]{\widehat M =\ob{M}\mid T=t}.
    \end{align*}
\end{proposition}
\begin{proof}
The joint density of $T$ and $W$ is given by $p_{\theta^{\ob{M}}}(t) \times p_W(w)$. The joint density conditional on $\widehat M(D, W)=\ob M$ is proportional to
\begin{align*}
    p(t, w\mid\widehat M(D,w)=\ob M )\propto p_{\theta^{\ob{M}}}(t) \cdot p_W(w) \cdot \PP{\widehat M(D,w) = \ob M \mid T=t, W=w }.
\end{align*}
Integrating over $w$, we obtain
\begin{align*}
    p_{\theta^{\ob{M}}}^*(t) \propto p_{\theta^{\ob{M}}}(t) \cdot \PP{\widehat M(D, W) = \ob M \mid T=t },
\end{align*}
where the expectation is taken over $W$ and $D\mid T=t$.
\end{proof}

\section{Details of the experiments}
\label{sec: experiment details}
\subsection{Normalizing flow architecture}

We construct the normalizing flow $\tau(\cdot)=\tau(\cdot;\psi)$ by stacking $L$ affine coupling layers \cite{dinh2017density}. Given a subset $u\subset1{:}d$, let $-u$ denote the complement of $u$. Let $\bfx$ and $\bfx'$ denote the input and output of an affine coupling layer, with the mapping given by
\begin{align*}
    \bfx_{u}' &= \bfx_u, \\
    \bfx_{-u}' &= \mathrm{softplus}(s(\bfx_u))\odot \bfx_{-u} + t(\bfx_u).
\end{align*}
The scaling function $s(\cdot)$ and shifting function $t(\cdot)$ are functions of $\bfx_u$, and are parametrized by multi-layer perceptrons (MLP). The softplus operator $\log(1+e^x)$ is used to ensure that the scaling is positive.

The coupling layer keeps the input variables $\bfx_u$ unchanged, while transforming the remaining variables $\bfx_{-u}$ using a componentwise affine transformation, whose scale and shift are determined by the values of $\bfx_u$. This structure ensures that the Jacobian of this transformation is a triangular matrix, and the determinant can be computed efficiently. The inverse of the transformation can be computed by applying the inverse of the affine transformation to $\bfx_{-u}'$ and keeping $\bfx_u$ unchanged.

For conditional normalizing flows, we concatenate the parameter value $\theta$ with the input $\bfx_u$ in every scaling and shifting function, and they become $s(\bfx_u, \theta)$ and $t(\bfx_u,\theta)$, respectively. In all the experiments, we use $L=12$ coupling layers. Each scaling and shifting MLP has one hidden layer with 8 neurons, and uses the ReLU activation.

For the one-dimensional case, we parametrize the one-dimensional map $\tau:\R\to\R$ by a monotonic rational-quadratic spline \cite{durkan2019neural}, where the knots and the derivatives at the internal points are the parameters to be learned. If a conditional flow is trained, the knots and derivatives are outputs of an MLP, whose input is the parameter $\theta$. 
In all the experiments, we use 20 bins for the rational-quadratic spline. 

\subsection{Training details}

We generate $2000$ training samples and $500$ validation samples. The training samples are used to compute the loss function in~\eqref{equ: kl objective}. We run full-batch Adam with 10000 iterations, and compute the validation loss using the validation samples every 1000 iterations. The parameter corresponding to the lowest validation loss is selected as the final parameter. The learning rate is set to $10^{-4}$ for all the experiments initially. If training diverges, the learning rate is changed to $10^{-5}$.

\subsection{Generating training data}

In all the experiments except the example of spline regression, we need to train a conditional normalizing flow $\hat\tau_\theta$ that pushes forward $\bbP_\theta$ to $\bbP_\theta^*$ for $\theta\in\Theta$. Therefore, we need to generate pairs of $(T^{(b)},\theta^{(b)})$ such that $T^{(b)}\sim \bbP_{\theta^{(b)}}^*$ for $1\leq b\leq B$. In the implementation, we draw $\theta^{(b)}$ from a normal distribution, whose center is the (unconditional) MLE of $\theta$ and whose covariance is the covariance of the MLE. Given every $\theta^{(b)}$, we draw $T^{(b)}\sim \bbP_{\theta^{(b)}}^*$ by rejection sampling. If the rejection sampling fails to produce a sample after 100 tries, we stop and continue to generate the next $\theta^{(b+1)}$.

\subsection{PBMC data preprocessing}
\label{sec: preprocessing details}
We use the \texttt{Seurat} package to preprocess the gene expression data. We first filtered out low-quality cells by retaining only those that expressed more than 200 and fewer than 5000 genes, and had less than 5\% of their total expression coming from mitochondrial genes. Gene expression counts were then normalized to account for differences in sequencing depth across cells. From the normalized data, we identified the 2000 most variable genes across all cells and standardized their expression levels to have zero mean and unit variance. Finally, we only keep the cells that are annotated as either ``memory B cell'' or ``naive B cell'', resulting in a data matrix of size $233\times 2000$. Principal components are computed based on this matrix.

\bibliographystyle{apalike}
\bibliography{ref}
\end{document}